\documentclass[]{article}
\usepackage[a4paper]{geometry}
\usepackage[utf8]{inputenc}
\usepackage{amsfonts}
\usepackage{amsmath,amstext,amsthm}
\usepackage{amssymb}
\usepackage{graphicx}
\usepackage{natbib}
\usepackage{tikz}
\usepackage{bbm}
\usepackage{comment}
\usepackage{setspace}
\usepackage[ruled,vlined,algo2e]{algorithm2e}
\usepackage{multirow}
\usepackage{rotating}
\usepackage{algorithm}
\usepackage{setspace}
\usepackage{verbatim}
\usepackage{dsfont}
\usepackage{cases}
\usepackage{float}
\newfloat{algorithm}{h}{lop}

\newtheorem{proposition}{Proposition}[section]
\newtheorem{corollary}{Corollary}[proposition]
\newtheorem{lemma}[proposition]{Lemma}
\newtheorem{theorem}[proposition]{Theorem}
\newtheorem{remark}[proposition]{Remark}
\newcommand{\indep}{\perp\!\!\!\!\!\!\perp} 
\newcommand{\1}{\mathds{1}}
\newcommand{\model}{\mathfrak{M}}

\title{Evidence estimation in finite and infinite mixture models and applications}
\author{Adrien Hairault${}^1$, Christian P.~Robert${}^{1,2}$ and Judith Rousseau${}^3$\\${}^1$Université Paris Dauphine, ${}^2$University of Warwick, and ${}^3$University of Oxford }

\date{\today}

\begin{document}

\maketitle
\begin{abstract}
    Estimating the model evidence - or mariginal likelihood of the data - is a notoriously difficult task for finite and infinite mixture models and we reexamine here different Monte Carlo techniques advocated in the recent literature, as well as novel approaches based on \cite{geyer1994estimating} reverse logistic regression technique, \cite{chib1995marginal} algorithm, and Sequential Monte Carlo (SMC). Applications are numerous. In particular,
    testing for the number of components in a finite mixture model or against the fit of a finite mixture model for a given dataset has long been and still is an issue of much interest, albeit yet missing a fully satisfactory resolution. Using a Bayes factor to find the right number of components $K$ in a finite mixture model is known to provide a consistent procedure. We furthermore establish the consistence of the Bayes factor when comparing a parametric family of finite mixtures against the nonparametric ‘strongly identifiable' Dirichlet Process Mixture (DPM) model. 
\end{abstract}
\section{Introduction}

Mixture models, defined as the convex combination of probability distributions, are of significant interest due to their convenient way of modelling heterogeneity in a given population. When some group structure is present in the data, Finite Mixture models (FM) arise as a natural tool to conduct a Bayesian clustering analysis. They have been successfully applied to computer vision (\cite{stauffer1999adaptive}), document classification (\cite{blei2003latent}) but also more generally to genetics, physics, or economics (see \cite{mclachlan2019finite} for a comprehensive review of applications). Alternatively, the Dirichlet Process Mixture model (DPM), or ‘infinite' mixture model, first introduced by \cite{ferguson1983bayesian}, has become one of the main tools in the field of Bayesian non-parametrics. Its range of applications is broad, from multivariate clustering (see, e.g., \cite{crepet2011bayesian}, \cite{quintana2003bayesian}) to Bayesian density estimation (\cite{escobar1995bayesian}, \cite{neal1992bayesian}, \cite{rasmussen1999infinite}, \cite{lo1984class}).

Bayesian model selection is primarily done by comparing the \textit{marginal likelihoods} of the data (a.k.a model evidence) for competing models, defined as
\begin{equation}
    \label{evidence}
m(y)=\int_{\Theta}f(x|\theta)\pi(d\theta)
\end{equation} for some data $y$, density function $f$, prior $\pi$, and parameter space $\Theta$.

In the context of mixture modelling, being able to compute the evidence of a model is of great importance and has several applications of interest. We here present two of them.

First, the usual problem of statistical modelling of determining whether an $n-$sample $( y_1,\dots, y_n)$ arises from a particular parametric family of distributions can be formalised as a goodness of fit test against a nonparametric alternative, like the DPM (\cite{kamary2014testing}). \cite{tokdar2021bayesian} designs a particular DPM alternative for testing normality. The Bayes Factor was proven to be consistent for a point null hypothesis against a large class of nonparametric alternatives by \cite{dass2004note} or \cite{verdinelli1998bayesian}. \cite{mcvinish2009bayesian} derive sufficient conditions on nonparametric distributions for which the consistence of the Bayes factor for testing a parametric family of distributions holds. We establish here (Theorem \ref{UBMDP1}) the consistence of the Bayes factor comparing the parametric family of finite mixtures against the nonparametric location Dirichlet Process Mixture (DPM) model. Obviously, these results can only be of use in practice when one is able to produce a numerical approximation to marginal likelihoods and, subsequently, to the Bayes Factor.

Second, marginal likelihood can be used to find the right number of components $K$ in a finite mixture model. One way of addressing this challenge is indeed to perform model selection using the Bayes Factor by computing the marginal likelihood of the data for various values of $K$. From a clustering perspective, it can be regarded as finding the optimal number of groups $K$ in a population.

%For finite mixtures, it can be used to infer the number of mixture components $K$, which is a quantity of great interest for clustering and good model specification. 

Evaluating the model evidence (\ref{evidence}) is unfortunately a notoriously difficult task for mixture models.
%Both finite and infinite mixtures require careful model specification in order to be valid tools for inference. This can be achieved by Bayesian model selection through the Bayes factor which allows to compare any two models without assuming any nested structure between them. It is defined as the ratio of the \textit{marginal likelihoods} of the data (a.k.a model evidence) for the competing models.
For finite mixtures, several Monte Carlo methods have been developed in order to estimate (\ref{evidence}), among which Importance Sampling (IS) methods such as Bridge Sampling are popular \citep{fruhwirth2004estimating}. \cite{chib1995marginal} proposed another approach based on the Bayes identity which initial bias, as pointed out by \cite{neal1999erroneous}, can be corrected \citep{lee2016importance}. However, if one wishes to retrieve an estimator with a reasonable variance, those methods get very computationally expensive as $K$, the number of mixture components, increases (see \cite{lee2016importance}). Values of $K$ as small as $6$ already represent a significant computational burden, stressing the need for new estimators more suited to real-life applications where $K$ need not be small.

While there exists an abundant literature and well-established methods for computing the evidence of parametric models, equivalent methods for non-parametric models are far from numerous. In fact, to our knowledge, only \cite{basu2003marginal} directly address the issue of evidence estimation for the DPM by adapting the method of \cite{chib1995marginal}. SMC tools have been proposed by \cite{maceachern1999sequential} for beta-binomial Dirichlet Process mixtures. This idea was later generalized by \cite{griffin2011sequential} and applied to a very particular kind of DPM by \cite{tokdar2021bayesian}. However, their proposed SMC framework relies on the strong assumption that the concentration parameter $M$ of the DP is known and fixed. \cite{quintana2000computational} give an \textit{ad hoc} procedure in order to find the maximum likelihood estimator of $M$. 
Neither Chib's algorithm nor SMC appear to be widely used by practitioners.
Hence, it is a common practice to use the DPM without considering alternative parametric or non-parametric models, although it may not always be appropriate.
For instance, \cite{miller2014inconsistency} show that for data assumed to come from a finite mixture with an unknown number of components, the DPM posterior on the number of clusters in the observed data does not converge to the true number of components.

Using the partitioning structure induced by finite mixture models, we derive in Section \ref{sec_1} an unbiased estimator of the marginal likelihood of a conjugate finite mixture model which computational time does not grow exponentially with $K$, as is the case with the popular \cite{chib1995marginal} method, for instance. Based on \cite{kong1994sequential} sequential imputation algorithm, we also propose an effective Sequential Importance Sampling (SIS) strategy for estimating the evidence of conjugate finite mixtures. Given their simplicity and effectiveness, we believe these marginal likelihood estimators should greatly ease tasks like the selection of the number of components for finite mixture models, for instance. We also provide a comparative assessment of many different marginal likelihood estimators, and identify those that scale well in difficult scenarios where both the size of the data $n$ and the number of components $K$ are large, which to our knowledge has not been done before.

As for the DPM, we propose in Section \ref{sec:DPM} an alternative algorithm based on Reverse Logistic Regression (\cite{geyer1994estimating}) which does not require the concentration parameter $M$ to be fixed but rather to follow a Gamma prior distribution. We show empirically that this method scales better with the amount of data than \cite{basu2003marginal} does. We also provide a review and assessment of the ways to estimate the marginal likelihood of a non-parametric DPM model, which, to our knowledge, has not been done before. In a similar way to Section \ref{sec_1}, we provide an empirical study of marginal likelihood estimators, including scenarios where $n$ is large. We also empirically assess the behaviour of the Bayes Factor comparing a parametric family of finite mixtures against a nonparametric DPM alternative.

\begin{comment}
For both the Dirichlet Process Mixture model and the Finite mixture model, it is of great interest to study the asymptotic behavior of the marginal likelihood. While such properties are now well-established for parametric models, asymptotic upper-bounds for the marginal likelihood of the DPM still need to be derived. A particularly interesting application is proving the consistence of the Bayes Factor for a parametric null hypothesis against a DPM alternative.
\end{comment}

While the asymptotic behaviour of the marginal likelihood has been studied In Section \ref{sec:theory}, we study the asymptotic behavior of the marginal likelihood of the DPM for strongly identifiable emission distributions, when fitted to data assumed to arise from a finite location mixture. The obtained result implies the consistence of the Bayes Factor in this setting. While the asymptotic behaviour of a finite mixture is well-established (see for instance \cite{chambaz2008bounds} or \cite{drton2017bayesian}), equivalent results for the DPM and in particular an upper bound on the marginal likelihood, had not been derived before.

\section{Evidence approximation for finite mixtures}
\label{sec_1}

Assume that a random \textit{i.i.d} sample $\boldsymbol{y}=( y_1,\dots, y_n)$ arises from a finite mixture of $K$ distributions which densities $f(\cdot|\theta_k),\; k=1\dots,K$ are parametrized by a component-specific $\theta_k\in\Theta\subset\mathbb R^d$, $d\geq 1$. One can write the corresponding mixture likelihood function as
\begin{equation}
    \label{finmix}
    p_K(\boldsymbol y|\boldsymbol{\vartheta})=\prod_{i=1}^n\sum_{k=1}^K \varpi_kf(y_i|\theta_k)
\end{equation}
where $\boldsymbol\vartheta=(\boldsymbol\theta,\boldsymbol\varpi)$, $\boldsymbol \theta=(\theta_1,\dots,\theta_K)\in\Theta^K\subset \mathbb R^{d\times K}$, $\boldsymbol \varpi=(\varpi_1,\dots,\varpi_K)\in\Delta_{K-1}$, the $(K-1)-$dimensional simplex. \\
Another way of specifying a finite mixture model is by introducing latent variables $\boldsymbol z=(z_1,\dots,z_n)\in\{1,\dots,K\}^n$ that indicate cluster membership of the individual observations $(y_1,\dots,y_n)$. They allow for a convenient generative construction of the finite mixture model.
\begin{align}
\label{latentvar}
    \nonumber y_i|\boldsymbol\theta, z_i&\overset{i.i.d.}{\sim} f(.|\theta_{z_i}),\; i=1,\dots,n\\
    \nonumber z_1,\dots,z_n|\boldsymbol\varpi&\overset{i.i.d.}{\sim}Cat(\varpi_1,\dots,\varpi_K)\\
    \boldsymbol\varpi=(\varpi_1,\dots,\varpi_K)&\sim \mathcal D(\boldsymbol\alpha)\\
    \nonumber\theta_1,\dots,\theta_K&\overset{i.i.d.}{\sim} G_0
\end{align}
where $G_0$ is a prior on $\theta$ while $\mathcal D(\cdot)$ and $Cat(\cdot)$ denote the Dirichlet and the categorical distributions, respectively.

\subsection{Popular existing estimators of the marginal likelihood for finite mixtures}

Since the seminal article by \cite{tanner1987calculation}, a common strategy for sampling from the posterior distribution of a finite mixture model $\pi_K(\boldsymbol\vartheta|\boldsymbol y)$ takes advantage of the latent variable representation (\ref{latentvar}). Popular Gibbs sampling schemes have been developed for this purpose by \cite{diebolt1994estimation}, for example. However, efficient and reliable algorithms to estimate the marginal likelihood of such models, that scale with the number of mixture components $K$, still need to be derived. One can define this quantity by 
 
% In a Bayesian model selection setting, finding the right number of components $K$ in a finite mixture model, although of crucial importance, is a notoriously difficult problem. 

%One possible and natural Bayesian way of addressing this challenge is to put a prior on the number of components $K$, resulting in a Mixture of Finite Mixtures (MFM) model. Though theoretically attractive, MFMs posterior distribution are usually difficult to sample from as the parameter space dimension varies with $K$. \cite{green1995reversible}'s Reversible Jump MCMC (RJMCMC) algorithm is still nowadays the most popular method to perform inference for MFMs but its good mixing properties heavily depend on the quality of the problem-dependent and user-defined reversible jump proposals. Note that new promising samplers such as the telescoping sampler of \cite{fruhwirth2020generalized} aim at giving more reliable tools for MFM inference.
\begin{equation}
    \label{marginallikelihood}
    m_K(\boldsymbol y):=\int_{\Theta^K\times \Delta_{K-1}}p_K(\boldsymbol y|\boldsymbol\vartheta)\pi_K (d\boldsymbol\vartheta)
\end{equation}
where $\pi_K(\boldsymbol \vartheta)=G_0(\boldsymbol\theta)\mathcal D(\boldsymbol \varpi|\boldsymbol\alpha)$.
Integral (\ref{marginallikelihood}) is unfortunately intractable and difficult to estimate, in particular because of the $K!$ equiprobable modal configurations of the posterior distribution. This specific complexity is at the source of the development of \textit{ad hoc} estimation methods for finite mixture models. Below we describe the two most popular Monte-Carlo algorithms for marginal likelihood estimation of FM, namely Bridge Sampling (\cite{fruhwirth2019keeping}) and Chib's algorithm (\cite{chib1995marginal}) as well as its subsequent corrections as proposed by \cite{berkhof2003bayesian} and \cite{lee2016importance}. We explain why these algorithms suffer from the curse of dimensionality as $K$ increases. Sequential Monte Carlo (SMC) is also described in its finite mixture version as suggested by \cite{chopin2002sequential}. We then propose another correction to Chib's algorithm based on the partitioning structure implied by mixture models. We also adapt \cite{kong1994sequential}'s sequential imputation strategy to finite mixture model, which appears to be both robust to an increase in the number of mixture components $K$ and the number of observations $n$.
%In a Bayesian model selection setting, finding the right number of components $K$ in a finite mixture model, although of crucial importance, is a notoriously difficult problem. 
%One way of addressing this challenge is to perform model comparison for various values of $K$, by computing the marginal likelihood of the data

\paragraph{Chib's algorithm.} Introduced by \cite{chib1995marginal}, the core idea of this method relies on the simple Bayes identity
\begin{equation}
\label{Chibidentity}
    m_K(\boldsymbol y)=\frac{p_K(\boldsymbol y|\boldsymbol\vartheta_0)\pi_K(\boldsymbol\vartheta_0)}{\pi_K(\boldsymbol\vartheta_0|\boldsymbol y)}
\end{equation}
for some $\boldsymbol\vartheta_0\in\Theta^{K}\times\Delta_{K-1}$.\\
To estimate the posterior density $\pi_K(\boldsymbol\vartheta_0|\boldsymbol y)$ that is not available in closed-form, \cite{chib1995marginal} uses the output $\{\boldsymbol\vartheta^{(t)},\boldsymbol z^{(t)}\}_{t=1}^T$ of a Gibbs sampler targeting the posterior distribution of ($\boldsymbol\vartheta,\boldsymbol z)$ in a Rao-Blackwellized estimator 
\begin{equation}
    \label{raoblackwell}
    \hat \pi_K(\boldsymbol\vartheta_0|\boldsymbol y)=\frac1T\sum_{t=1}^T \pi_K(\boldsymbol\vartheta_0|\boldsymbol z^{(t)},\boldsymbol y)
\end{equation}
where $\boldsymbol\vartheta_0$ typically is the Maximum A Posteriori (MAP) estimator. Note that the densities $\pi_K(\cdot|\boldsymbol z,\boldsymbol y)$ are available in closed form for conjugate finite mixtures.
The posterior density estimator (\ref{raoblackwell}), that is plugged into (\ref{Chibidentity}), is unfortunately biased in the case of poor mixing of the Monte Carlo Markov Chain producing the sequence $\{\boldsymbol\vartheta^{(t)},\boldsymbol z^{(t)}\}_{t=1}^T$, as pointed out by \cite{neal1999erroneous}. Indeed, due to mixture posterior invariance to relabeling of components (see, e.g., \cite{fruhwirth2006practical}), there exists $K!$ equally likely posterior modes. When the Gibbs sampler does not visit evenly every one of the $K!$ modes, or in other words all $K!$ permutations of the MAP estimator $\boldsymbol z_0$, (\ref{raoblackwell}) is over-estimated, which in turn implies an underestimation of Chib's marginal likelihood estimator. This ill-behavior of the Gibbs sampler is in practice very common.
In the extreme case where the sampler only visits the mode represented by $\boldsymbol z_0$, Chib's estimator underestimates the marginal likelihood ‘exactly' by a factor $K!$. \\
For small values of $K$, however, an easy fix has been suggested by \cite{berkhof2003bayesian}. It consists in summing over all permutations of $\boldsymbol z^{(t)}$ for $t=1,\dots,T$ yielding the following estimator 
\begin{equation}
    \label{berkhof}
    \tilde\pi_K(\boldsymbol\theta_0|\boldsymbol y)=\frac{1}{K!T}\sum_{t=1}^T\sum_{s\in\mathfrak S(\{1,\dots,k\})}\pi_K(\boldsymbol\vartheta_0|s(\boldsymbol z^{(t)}),\boldsymbol y)
\end{equation}
where $\mathfrak S(\{1,\dots,k\})$ denotes the set of all permutations of $\{1,\dots,k\}$. \\
This estimator, although very accurate, cannot be computed in an appropriate time as soon as $K\geq 5$. A trick suggested by \cite{marin2008approximating} is to sample at random $100$ permutations for such values of $K$ in order to keep computational time below a reasonable threshold. Unfortunately, 100 random permutations should not be enough to keep the good variance properties of Chib's estimator when $K\geq8$, as for such values, $K!\gg 100$.

\paragraph{Bridge Sampling.} Bridge sampling is a very popular generalisation of Importance Sampling introduced by \cite{meng1996simulating}. It relies on identity \eqref{eq:BSidentity} that holds for all positive function $\alpha(\boldsymbol \vartheta)$ such that $\int_\Theta \alpha(\boldsymbol \vartheta) q_K(\boldsymbol\vartheta)\pi_K(\boldsymbol\vartheta|\boldsymbol y)d\boldsymbol\vartheta>0$, where $q_K(\boldsymbol\theta)$ is an importance density aiming at approximating the posterior distribution $\pi_K(\boldsymbol\vartheta|\boldsymbol y)$,
\begin{equation}\label{eq:BSidentity}
    \int_\Theta \alpha(\boldsymbol \vartheta) q_K(\boldsymbol\vartheta)\pi_K(\boldsymbol\vartheta|\boldsymbol y)d\boldsymbol\vartheta=\int_\Theta \alpha(\boldsymbol \vartheta) q_K(\boldsymbol\vartheta)\frac{p_K(\boldsymbol y|\boldsymbol\vartheta)\pi_K(\boldsymbol\vartheta)}{m_K(\boldsymbol y)}d\boldsymbol\vartheta
\end{equation}
which gives the Bridge Sampling estimator of the marginal likelihood 

\begin{equation*}
    m_K(\boldsymbol y)=\frac{\mathbb E_{q_K(\boldsymbol\vartheta)}(\alpha(\boldsymbol \vartheta)p_K(\boldsymbol y|\boldsymbol\vartheta)\pi_K(\boldsymbol\vartheta))}{\mathbb E_{\pi_K(\boldsymbol \vartheta|\boldsymbol y)}(\alpha(\boldsymbol\vartheta)q_K(\boldsymbol\vartheta))}
\end{equation*}

Assuming one can obtain $T_1$ i.i.d draws $\boldsymbol\vartheta_1^{q},\dots,\boldsymbol\vartheta_{T_1}^{q}$ from $q_K(\boldsymbol\theta)$ and $T_2$ (auto-correlated) MCMC samples $\boldsymbol\vartheta_1^{\pi},\dots,\boldsymbol\vartheta_{T_2}^{\pi}$ from $\pi_K(\boldsymbol\vartheta|\boldsymbol y)$ respectively, \cite{meng1996simulating} derive the optimal choice for $\alpha(\boldsymbol\vartheta)$ as

\begin{equation*}
    \alpha(\boldsymbol\vartheta)=(N_1q_K(\boldsymbol\vartheta)+N_2^\ast\pi_K(\boldsymbol\vartheta|y))^{-1}
\end{equation*}
with $N_2^*$ the Effective Sample Size (ESS) of the MCMC sample from the posterior. As this optimal choice of $\alpha(\boldsymbol\vartheta)$ involves the unknown marginal likelihood $m_K(\boldsymbol y)$, the BS estimator is written recursively as 

\begin{equation}
    \label{BS}
    \hat m_{K,n}^{BS}(\boldsymbol y)=\frac{\sum_{t=1}^{T_1}\frac{p_K(\boldsymbol y|\boldsymbol\vartheta_t^q)\pi_K(\boldsymbol\vartheta_t^q)}{N_1q_K(\boldsymbol\vartheta_t^q)+N_2^*\frac{p_K(\boldsymbol y|\boldsymbol\vartheta_t^q)\pi_K(\boldsymbol\vartheta_t^q)}{\hat m_{K,n-1}^{BS}(\boldsymbol y)}}}{\sum_{t=1}^{T_2}\frac{q_K(\boldsymbol\vartheta_t^\pi)}{N_1q_K(\boldsymbol\vartheta_t^\pi)+N_2^*\frac{p_K(\boldsymbol y|\boldsymbol\vartheta_t^\pi)\pi_K(\boldsymbol\vartheta_t^\pi)}{\hat m_{K,n-1}^{BS}(\boldsymbol y)}}}
\end{equation}
where $m_{K,0}^{BS}$ is usually a simple IS estimator with $q(\boldsymbol\vartheta)$ as importance distribution.
As $n$ grows, one can retrieve the optimal BS estimator $m_{K}^{BS}(\boldsymbol y)$ as $\underset{n\rightarrow\infty}{\lim}\hat m_{K,n}^{BS}(\boldsymbol y)= m_{K}^{BS}(\boldsymbol y)$.

\cite{fruhwirth2019keeping} provides a comprehensive review on how to successfully apply the Bridge Sampling framework to finite mixtures. As can be expected, the difficulty lies in deriving a good importance density $q_K(\boldsymbol\vartheta)$ that yields a good approximation to the notoriously complex posterior distribution of a finite mixture model. In the same vein as \cite{berkhof2003bayesian}'s permuted Chib's estimator (\ref{berkhof}), \cite{fruhwirth2019keeping} proposes the following choice for $q(\boldsymbol\vartheta)$:

\begin{equation*}
    q(\boldsymbol\vartheta)=\frac{1}{K!T_0}\sum_{t=1}^{T_0}\sum_{s\in\mathcal \sigma(\{1,\dots,k\})}\pi_K(\boldsymbol\vartheta|s(\tilde{\boldsymbol z}^{(t)}),\boldsymbol y)
\end{equation*}
where $T_0\ll T_1$ and $\{\tilde{\boldsymbol z}^{(t)}\}_t^{T_0}$ is drawn with replacement from the Markov chain targeting the posterior $\{\boldsymbol z^{(t)}\}_t^{T_1}$. Although $T_0$ is a smaller number than $T_1$, it is clear that this estimator suffers from the same type of computational shortcomings as Chib's symmetrized estimator. For a given number of observations $n$, computing estimator (\ref{BS}) comes at the cost of evaluating $T_0\times K!\times(T_1+ T_2)$ times the augmented posterior $\pi_K(\boldsymbol\vartheta|\boldsymbol z,\boldsymbol y)$ distribution.

Although both Chib's and Bridge Sampling approaches are well-established and known to provide reliable estimators of the marginal likelihood of finite mixtures, their computational time grows exponentially with the number of components $K$.

%As argued by \cite{berger1996intrinsic}, the Bayesian approach to model selection through the Bayes factor is a powerful paradigm. It naturally incorporates the Occam's razor principle in that it favors the simpler of two equally-informative models. It is also very flexible as it allows to compare any models, without assuming any particular nested structure. However, computing a Bayes factor is not straightforward and requires a good approximation of the usually intractable and multi-dimensional integral of the likelihood over the prior distribution called \textit{marginal likelihood} (a.k.a model evidence) of the data. The Bayes factor comparing two models $A$ and $B$ is indeed defined as the ratio of the marginal likelihoods of models $A$ and $B$.
\paragraph{Sequential Monte-Carlo (SMC).} Formalised by \cite{del2006sequential}, SMC is based on Sequential Importance Sampling (SIS). By sequentially sampling from a collection of distributions $\{\pi_t\}_{t=0}^T$ where $\pi_0$ is usually the prior distribution and $\pi_T$ is the posterior distribution of interest, the core idea of SMC is to effectively create a bridge from the prior to the posterior distribution. At each step $t=1,\dots,T-1$ of the algorithm, weighted samples $\{\boldsymbol\vartheta^{(i)}_t,w^{(i)}_t\}_{i=1}^N$ from $\pi_t$ are reweighted through an importance sampling step so that their distribution becomes approximately $\pi_{t+1}$.
\begin{equation}\label{eq:SMCweight}
    w_{t+1}^{(i)}=\frac{\pi_{t+1}(\boldsymbol\vartheta^{(i)}_t)}{\pi_{t}(\boldsymbol\vartheta^{(i)}_t)}
\end{equation}
Equation \eqref{eq:SMCweight} above stresses the importance of choosing successive distributions $\pi_t$ and $\pi_{t+1}$ as not being too dissimilar in order to ensure good variance properties. One common strategy is to choose the sequence of distributions $\{\pi_t\}_{t=0}^T$ to be a tempered version of the posterior distribution. More precisely, a tempering sequence $(\lambda_t)_t$ is leading to the ‘tempered' posterior distributions  
\begin{equation}\label{eq:SMCtemperatures}
    \pi_t(\boldsymbol\vartheta):=\pi_t(\boldsymbol\theta|y)\propto \pi_K(\boldsymbol\vartheta)p_K(\boldsymbol y|\boldsymbol \vartheta)^{\lambda_t}
\end{equation}
where $$\lambda_0=0<\lambda_1<\dots<\lambda_T=1$$
Based on readily available estimates of the Effective Sample Size (ESS), \cite{buchholz2021adaptive} gives a fully adaptive method to choose the tempering sequence $(\lambda_t)_t$. Another alternative is to use data-tempered posterior distributions in which batches of data are sequentially incorporated into the posterior. This approach is closely related to the Sequential Importance Sampling (SIS) strategy we suggest in the next sub-section. Therefore, we here use the tempering strategy given in \eqref{eq:SMCtemperatures} so that both methods are described in this article.\\
A multinomial resampling step is usually implemented in order to select the particles with relatively higher weights \eqref{eq:SMCweight}. This unfortunately can cause particle degeneracy that occurs when only a few particles $\{\boldsymbol\vartheta^{(i)}_t\}$ with a large weight are selected during the resampling step. To compensate for this frequent shortcoming, a \textit{mutation} step is added to the algorithm in which each particle is moved through an MCMC $\pi_{t+1}$-invariant kernel $K_t$. In practice, kernel $K_t$ is applied to each particles $M$ number of times to ensure good mixing, where $M$ is a user-defined number.\\
Although SMC algorithms are known to be powerful in that they are robust to the curse of dimensionality, all these tuning steps make SMC a tedious method to implement. Following the successful application of SMC to finite mixture models of \cite{chopin2002sequential}, we make use of a random walk Metropolis-Hasting kernel.

As shown in \cite{del2006sequential}, an unbiased estimate of the marginal likelihood $m_K(\boldsymbol y)$ is given as a by-product of SMC as

\begin{equation}
    \label{SMCmarginalLlk}
    \hat m^{SMC}_K(\boldsymbol y)=\prod_{t=0}^T\frac{1}{N}\sum_{i=1}^N w_t^{(i)}
\end{equation}

\begin{algorithm}[H]
\setstretch{1.35}
\label{algoSMCFM}
\caption{\textit{Tempered SMC}}
\textbf{Input} : Number of particles $N$, prior distribution $\pi_K$, Markov kernels $K_t$ that are $\pi_{t}$ invariant, where $\pi_{t}(\boldsymbol\vartheta|\boldsymbol y)\propto \pi_K(\boldsymbol\vartheta)p_K(\boldsymbol y|\boldsymbol\vartheta)^{\lambda_t}$ \\
\textbf{Initialisation} : $t=0$, $\lambda_0=0$\\
\While{$\lambda_t<1$}{
    \If{t=0}{
        \For {$i=1,\dots,N$}{
            Sample $\boldsymbol\vartheta_0^{(i)}\sim\pi_K$
        }
        
    }
    \Else{
        \For {$i=1,\dots,N$}{
            \textit{Mutation.} Move particles $\boldsymbol\vartheta_t^i\sim K_{t}(\Tilde{\boldsymbol\vartheta}_{t-1}^{(i)},d\boldsymbol\vartheta)$ (Repeat $M$ times)
        }
       
    }
    Find the next temperature $\lambda_{t+1}>\lambda_{t}$ adaptively using the current state of the particules.
    
    \textit{Reweighting.} 
    \For{$i=1,\dots,N$}{
        $w_{t+1}^{(i)}=p_K(\boldsymbol y|\boldsymbol \vartheta_t^{(i)})^{\lambda_{t+1}-\lambda_{t}}$  
    } 
    \textit{Resampling.} Resample with replacement from $\{\boldsymbol\vartheta_{t+1}^{(i)},w_{t+1}^{(i)}\}_{i=1}^n$ to obtain a new sample with equal weights $\{\Tilde{\boldsymbol\vartheta}_{t+1}^{(i)},\frac 1N\}$
    
    $t=t+1$
    
}
\textbf{Output} : $\hat m^{SMC}_K(\boldsymbol y)=\prod_{t}\frac{1}{N}\sum_{i=1}^N w_t^{(i)}$
\end{algorithm}

\subsection{Proposed estimators}
\label{sec:novelEstimFM}
In this section, we present two novel estimators of the marginal likelihood for conjugate finite mixture models. One of them is inspired by Chib's method and takes advantage of the partitioning of the data induced by such models. It provides efficient and robust results as well as a dramatic reduction in computational time even for $K>5$. The second one is an application of \cite{kong1994sequential}'s sequential imputation algorithm, which is surprisingly not widely used as a solution to the problem of marginal likelihood estimation for finite mixtures.
\paragraph{Chib's estimator on the partition.}
Partitioning comes as a natural by-product of a classical Gibbs Sampler for finite mixture models, when the additional cluster membership latent variable $\boldsymbol z$ is introduced. Indeed, if the output of such an MCMC algorithm is $\{\boldsymbol\vartheta^{(t)},\boldsymbol z^{(t)}\}_{t=1}^T$, we denote by $\mathcal C(z^{(t)})$ the partition on $[n]=\{1,\dots,n\}$ induced by $z^{(t)}=(z_1^{(t)},\dots,z_n^{(t)})$, for all $t=1,\dots,T$.\\
For example, if $n=4$ and $\boldsymbol z=(1,2,1,3)$, then the corresponding partition is $\mathcal C(\boldsymbol z)=\{\{y_1,y_3\},\{y_4\},\varnothing,\{y_2\}\},$, i.e observations $y_1$ and $y_3$ are in the same cluster whereas $y_2$ and $y_4$ are the only members of their respective cluster.

The core idea of the {ChibPartition} estimator that we propose in this article, is to apply the marginal likelihood identity (\ref{Chibidentity}) to an estimator of the MAP partition. This reads

\begin{equation}
    \label{chibpartitions}
        p(y)=\frac{p_K(y|\mathcal C^0)\pi_K(\mathcal C^0)}{\pi_K(\mathcal C^0|y)},\quad \mathcal C^0\in\mathcal P_K([n])
\end{equation}
where we denote by $\mathcal P_K([n])$ the number of partitions of $\{1,\dots,n\}$ with at most $K$ parts.\\
Note that the induced prior on the partitions $\pi_K(\mathcal C)$ is available in closed form provided that a Dirichlet prior is used for the weights $\boldsymbol\varpi$. It can indeed be written as 
\begin{align}
\label{priorpartition}
    \nonumber\pi_K(\mathcal C(\boldsymbol z))&=\frac{K!}{(K-K_+)!}\pi_K(\boldsymbol z)\\
    \nonumber&=\frac{K!}{(K-K_+)!}\int \pi_K(\boldsymbol z|\boldsymbol\pi) d\mathcal D(\boldsymbol\pi|\boldsymbol\alpha)\\
    &=\frac{K!}{(K-K_+)!}\frac{\Gamma\left(\sum_{j=1}^K\alpha_j\right)}{\Gamma\left(\sum_{j=1}^K\alpha_j+n\right)}\prod_{j=1}^{K}\frac{\Gamma(n_j+\alpha_j)}{\Gamma(\alpha_j)}
\end{align}
where $n_j=\sum_{i=1}^n \1_{\{z_i=j\}}$, i.e the number of observations assigned to cluster $j$ for all $j=1,\dots,K$, and $K_+=\sum_{j=1}^K \1_{\{n_j>0\}}$, which is the number of non-empty clusters implied by a given allocation vector $\boldsymbol z$. 

Now for a conjugate prior $G_0$ on $\boldsymbol\theta$, the likelihood of a partition is available in closed form and given by
\begin{align}
\label{likelihoodpartition}
     p(\boldsymbol y|\mathcal C(\boldsymbol z))=p(\boldsymbol y|\boldsymbol z)=\prod_{j=1}^K{\int_\Theta\prod_{i: z_i=k}p(y_i|\theta)G_0(d\theta)}:=\prod_{j=1}^Km({C_k(\boldsymbol z)})
\end{align}
with the convention that $m(\varnothing)=1.$

The posterior density of a partition $\mathcal C$ is unfortunately not available in closed form. However, we can estimate it with the following simple Monte Carlo estimator readily computable from the Gibbs sampler output.

\begin{equation}
    \label{posteriorpartition}
    \hat\pi_K(\mathcal C^0|\boldsymbol y)=\frac1T\sum_{t=1}^T\1_{\{C^0\doteq C(z^{(t)})\}}
\end{equation}
for some sample $\{z^{(t)}\}$ distributed according to the posterior distribution and where we use $\doteq$ to denote an equivalence relation between two partitions. For example, $$\{\{y_1,y_3\},\{y_2\},\{y_4\}\}\doteq \{\{y_2\},\{y_4\},\{y_1,y_3\}\}$$ as both partitions imply the same partitioning of the data, up to a permutation of the clusters' index.
From a computational point of view, note that comparing two partitions with the equivalence relation $\doteq$ is a simple $\mathcal O(n\log n)$ operation. Below is given the pseudocode of the ChibPartition estimator.
\begin{algorithm}
\setstretch{1.35}
\caption{\textit{ChibPartition estimator}}\label{algochibpartition}
\textbf{Input} : $(\boldsymbol z^{(t)})_{t=1}^T$ from an MCMC at stationarity targeting $\pi_K(\boldsymbol z|\boldsymbol y)$ \\
\For {$t=1,\dots,T$}{
    $\text{Compute } \tilde\pi(\mathcal C(\boldsymbol z^{(t)})|\boldsymbol y)= p_K(y|\mathcal C(\boldsymbol z^{(t)}))\pi_K(\mathcal C(\boldsymbol z^{(t)}))$ using (\ref{likelihoodpartition}) and (\ref{priorpartition})
    
    $C^0=\underset{t=1,\dots,T}{\text{argmax }}\{\tilde\pi(\mathcal C(\boldsymbol z^{(t)})|\boldsymbol y)\}$
}
Set $\hat\pi(\mathcal C^0|\boldsymbol y)=0$\\
\For {$t=1,\dots,T$}{
    $\text{Compute } \hat\pi(\mathcal C^0|\boldsymbol y)=\hat\pi(\mathcal C^0|\boldsymbol y)+ \1_{\{C^0\doteq C(z^{(t)})\}}/T$
}
\textbf{Output} : $m_K(\boldsymbol y)={p_K(\boldsymbol y|\mathcal C^0)\pi_K(\mathcal C^0)}/{\hat\pi(\mathcal C^0|\boldsymbol y)}$
\end{algorithm}

Possible shortcomings of Algorithm \ref{algochibpartition} stem from the cardinality of the set $\mathcal P_K([n])$ potentially making the estimation of the posterior density difficult. Note that there exists no simple expression giving the exact value of $|\mathcal P_K([n])|$ for all $k$ and $n$ but one can write
\begin{equation}
\label{cardinality}
    |\mathcal P_K([n])|=\sum_{k=1}^K S(n,k):=\sum_{k=1}^K\frac{1}{k!}\sum_{l=0}^k(-1)^l {k\choose l}(k-l)^n
\end{equation}
where $S(n,k)$ are the Stirling numbers of the second kind, counting the number of ways to partition a set of $n$ distinguishable objects into $k$ nonempty subsets (see \cite{graham1989concrete} for example).
The number of possible partitions given by $(\ref{cardinality})$ is potentially very large. For instance, when using a finite mixture with $K=8$ components on $n=82$ data points, $|\mathcal P_K([n])|\approx 2.80\times10^{69}$. It is therefore crucial to choose $C^0$ to be the estimated MAP or any other high-posterior probability partition to retrieve good variance for estimator (\ref{posteriorpartition}). 
As simulations in Section \ref{secsimulationsFM} tend to show, such a strategy appears to be sufficient to ensure a robust estimator of the posterior density.

Compared to the other corrections of \cite{chib1995marginal}'s original method that have been proposed in the past such as \cite{berkhof2003bayesian}'s \textit{Chib Permutation Estimator} ({‘ChibPermut'} hereafter) or \cite{marin2008approximating}'s \textit{Chib Random Permutation Estimator} ({‘ChibRandPermut'} hereafter), our proposed solution has a considerably lower computational time. Note that those three methods only differ in the way they estimate the posterior density function at a given point in the parameters' space. Therefore it is enough to compare the computational cost of this estimation as is given by Table \ref{tablecomputationaltime}.

\begin{table}[H]
    \centering
    \begin{tabular}{c|c}
    &Complexity of posterior estimation\\
    \hline
      {ChibPermut}   & $\mathcal O(TK!)$ \\
      {ChibRandPermut}   & $\mathcal O(TR)$\\ 
      {ChibPartition} & $\mathcal O(T)$\\
    \end{tabular}
    \caption{Computational complexity of estimating the posterior density function for three corrections of \cite{chib1995marginal}'s original marginal likelihood estimator: $T$ is the length of the Markov Chain targeting the posterior distribution, $K$ is the number of mixture components, and $R$ is the number of random permutations chosen for the \textit{ChibRandPermut} estimator.}
    \label{tablecomputationaltime}
\end{table}

Thanks to the ergodicity of the Gibbs output and using the auto-correlation consistent variance estimator given by \cite{newey1986simple}, an estimate of the variance of (\ref{posteriorpartition}) is given by 
$$\hat{V}(\hat\pi_K(\mathcal C^0|\boldsymbol y))=\frac 1T\left( \sigma_0+2\sum_{s=1}^q\left(1-\frac{s}{q+1}\right) \sigma_s\right)$$
where
$$\sigma_s=\frac 1T\sum_{t=s+1}^T\left(\1_{\{C^0\doteq C(z^{(t)})\}}-\hat\pi_K(\mathcal C^0|\boldsymbol y)\right)^2$$
for $q$ large enough.\\
The Delta method immediately yields the following estimator for the variance of the marginal likelihood ChibPartition estimator.
$$\hat V(\log \hat m_K(\boldsymbol y))=\hat V(\log \hat\pi_K(\mathcal C^0|\boldsymbol y))=(\hat\pi_K(\mathcal C^0|\boldsymbol y))^{-2}{\hat V}(\hat\pi_K(\mathcal C^0|\boldsymbol y))\,.$$
\paragraph{Sequential Importance Sampling (SIS).} \cite{kong1994sequential} addresses the issue of missing data problems by sequential imputation, using a latent variable $\boldsymbol z$, representing the missing part of the data. The authors show that for a particular choice of importance distribution $\pi^*$, a Sequential Importance Sampling (SIS) procedure yields a direct estimator of the marginal likelihood $m(\boldsymbol y)$. This is applicable whenever one can sample easily from distributions $p(z_i|\boldsymbol y_{1:i},\boldsymbol z_{1:i-1})$ for all $i\geq 2$ where $\boldsymbol z_{1:i}=(z_1,\dots,z_i)$ and whenever the prequential predictive densities $p(y_i|\boldsymbol y_{1:i-1},\boldsymbol z_{1:i-1})$ are available in closed form for all $i\geq 2$. 
Indeed, let us define $\pi^*$ as an approximation to the posterior distribution in which the latent variable $\boldsymbol z$ is imputed sequentially,
\begin{equation*}
    \label{posteriorImputation}
    \pi^*(z_1,\dots,z_n|\boldsymbol y)=p(z_1|y_1)\prod_{i=2}^n p(z_i|\boldsymbol y_{1:i},\boldsymbol z_{1:i-1})
\end{equation*}
Then,
\begin{align*}
    \frac{\pi(z_1,\dots,z_n|\boldsymbol y)}{\pi^*(z_1,\dots,z_n|\boldsymbol y)}&=\frac{p(\boldsymbol z,\boldsymbol y)}{p(\boldsymbol y)}\frac{p(y_1)}{p(z_1,y_1)}\frac{p(z_1,y_1,y_2)}{p(z_1,y_1,z_2,y_2)}\dots\frac{p(\boldsymbol z_{1:n-1},\boldsymbol y)}{p(\boldsymbol z,\boldsymbol y)}\\
    &= \frac{w(\boldsymbol z,\boldsymbol y)}{p(\boldsymbol y)}
\end{align*}
where $w(\boldsymbol z,\boldsymbol y):=p(y_1)\prod_{i=2}^n p(y_i|\boldsymbol z_{1:i-1},\boldsymbol y_{1:i-1})$.\\
Hence, if $\{\boldsymbol z^{(t)}\}_{t=1}^T$ are drawn sequentially from $\pi^*$, the following quantity is an unbiased estimator of the marginal likelihood of the data $p(\boldsymbol y)$ 
\begin{equation}
    \label{SISestimator}
    \hat m_K^{SIS}(\boldsymbol y)=\frac{1}{T}\sum_{i=1}^Tw(\boldsymbol z^{(t)},\boldsymbol y)
\end{equation}
The link to the latent cluster membership variable $\boldsymbol z$ defined in (\ref{latentvar}) is direct and it is straightforward to derive the necessary quantities to apply the idea of \cite{kong1994sequential} to mixture models, as also noticed by \cite{carvalho2010particle}. Despite the well-known efficiency of such samplers, note that SIS, and, more generally, SMC approaches, are not as popular as Chib's algorithm or Bridge sampling, the latter being the favoured method of the STAN community.
However, for conjugate models, sampling from $\pi^*$ and evaluating $w(\boldsymbol z,\boldsymbol y)$ is simple (see details in Appendix \ref{sec:appendixSISFM}).

As remarked by \cite{irwin1994sequential}, a simple application of the Delta method yields the following estimator for the standard deviation of $\log \hat m^{SIS}_K(\boldsymbol y)$.
$$\hat {sd}(\log \hat m^{SIS}_K(\boldsymbol y))=\frac{1}{\sqrt{M}}\frac{\hat {sd}(w)}{\hat m^{SIS}_K(\boldsymbol y)}$$
where $\hat{sd}(w)$ denotes the sample standard deviation of $\{w(\boldsymbol z^{(t)},\boldsymbol y)\}_t$.
\begin{algorithm}[H]
\setstretch{1.35}
\caption{\textit{SIS}}\label{algoSMCFM}
\textbf{Input} : Number of iterations $T$ \\
\For{t=1,\dots,T}{
    Sample $z_1^{(t)}\sim Cat(\alpha_1/\sum_{j=1}^K \alpha_j,\dots,\alpha_K/\sum_{j=1}^K \alpha_j$)\\
    Compute $p(y_1)=m(\{y_1\})$\\
    Set $w^{(t)}\leftarrow p(y_1)$\\
    \For{i=2,\dots,n}{
        \For{k=1,\dots,K}{
            Compute $\gamma_k=\frac{m(C_k(\boldsymbol z_{1:i-1}^{(t)})\cup \{y_i\})}{m(C_k(\boldsymbol z_{1:i-1}^{(t)}))}\frac{n_k(\boldsymbol z_{1:i-1}^{(t)})+\alpha_k}{\sum_{j=1}^K n_j(\boldsymbol z_{1:i-1}^{(t)})+\alpha_j}$\\
        }
        Sample $z_i^{(t)}\sim Cat(\gamma_1/\sum_{k=1}^K \gamma_k,\dots,\gamma_K/\sum_{k=1}^K \gamma_k)$\\
        Set $w^{(t)}\leftarrow w^{(t)}\sum_{k=1}^K\gamma_k$
    }
}
Return $1/T\sum_{t=1}^Tw^{(t)}$
\end{algorithm}

\subsection{Simulation study}
\label{secsimulationsFM}
%We try our method on the \texttt{galaxy} data for $K=3,5,6,8$, with a normal mixing kernel and a conjugate prior (normal inverse gamma) on the mean and variance parameters. The weights $\pi$ have a Dirichlet $\mathcal D(1,\dots,1)$ prior.\\
%We estimate 20 times the marginal likelihood with vanilla Chib, permutation Chib and Chib on the partitions.\\
%As expected vanilla Chib is underestimating the marginal by a factor of about $\log(3!)$ for $K=3$. Our method competes well with permutation Chib and seems to have less variance.\\
\subsubsection{Experiment 1 : Galaxy data.} The first experiment we design aims at assessing the relative performance of our suggested ChibPartition and SIS algorithms in a basic setting. As is usually done in the mixture modelling community, we use the benchmark \texttt{galaxies} data set that contains the radial velocity of 82 galaxies, was introduced by \cite{postman1986probes}. A long series of articles addressing the issue of evidence computation in mixture models usually evaluate their results on this data (see e.g. \cite{chib1995marginal}, \cite{richardson1997bayesian}, \cite{lee2016importance}, \cite{fruhwirth2019keeping}). 

A conjugate mixture of normal distributions is implemented. We use the \textit{conditionally conjugate} normal-inverse gamma prior for the location and scale parameters $\boldsymbol\theta=(\boldsymbol\mu,\boldsymbol\sigma^2)$ as described in \cite{fruhwirth2006finite} that yields for all $k=1,\dots,K$
\begin{equation}
    \begin{aligned}
    \label{condConjPrior}
    \sigma_k^2&\sim \Gamma^{-1}(a_0,b_0)\\
    \mu_k|\sigma_k^2&\sim\mathcal N(\mu_0,\sigma_k^2/\lambda_0)
    \end{aligned}
\end{equation}
where $\Gamma^{-1}$ is the inverse gamma distribution in the shape and scale parametrization. The hyperparameters ($a_0,b_0,\mu_0,\lambda_0)$ are derived empirically following recommendation from \cite{raftery1996hypothesis} : $a_0=1.28,b_0=0.36(\overline{y^2}-\overline y^2),\mu_0=\bar y,1/\lambda_0=(y_{max}-y_{min})/2.6$. Note that this choice of prior ensures that $\pi_K(\boldsymbol\vartheta|\boldsymbol y,\boldsymbol z):=p(\boldsymbol\mu,\boldsymbol\sigma^2|\boldsymbol y,\boldsymbol z)$ is available in closed form, which is a prerequisite to most of the algorithms we wish to implement. 

Figure \ref{fig:boxplotFMgalaxy} shows boxplots of the estimates of the marginal likelihood given by the different methods described in the previous section for an increasing number of mixture components $K$. In each scenario, the simple arithmetic mean estimator computed with a very large number of simulations is included. It is defined as
\begin{equation}
    \label{eq:arithmeticMean}
    \hat m_K^{AM}(\boldsymbol y)=\frac1T\sum_{t=1}^T p_k(\boldsymbol y|\boldsymbol \vartheta^{(t)})
\end{equation}
given a sample $\{\boldsymbol\vartheta\}_{t=1}^T$ from the prior $\pi_K(\boldsymbol\vartheta)$. This prohibitively time-consuming estimator is here solely for benchmark purposes and used as a reference for the other methods. The harmonic mean estimator is also included in this comparative study. It is defined as
\begin{equation}
    \label{eq:harmonicMean}
    \hat m_K^{HM}(\boldsymbol y)=\left\{\frac1T\sum_{t=1}^T 1/p_k(\boldsymbol y|\boldsymbol \vartheta^{(t)})\right\}^{-1}
\end{equation}
given a sample $\{\boldsymbol\vartheta\}_{t=1}^T$ from the posterior $\pi_K(\boldsymbol\vartheta|\boldsymbol y)$.

Except for the arithmetic mean, all other algorithms are allocated as much time as required for them to converge, provided this time is reasonable. For instance, Bridge Sampling and the fully-permuted Chib's estimator (ChibPerm) are not included for $K=6$ and $K=8$ as they fail to converge in a time comparable to the other methods. Hence, Figure \ref{fig:boxplotFMgalaxy} only provides insights about which of the methods are able to provide reliable estimates for an increasing number of mixture components $K$.

\begin{sidewaysfigure}[]
        \centering
        \includegraphics[scale=0.4]{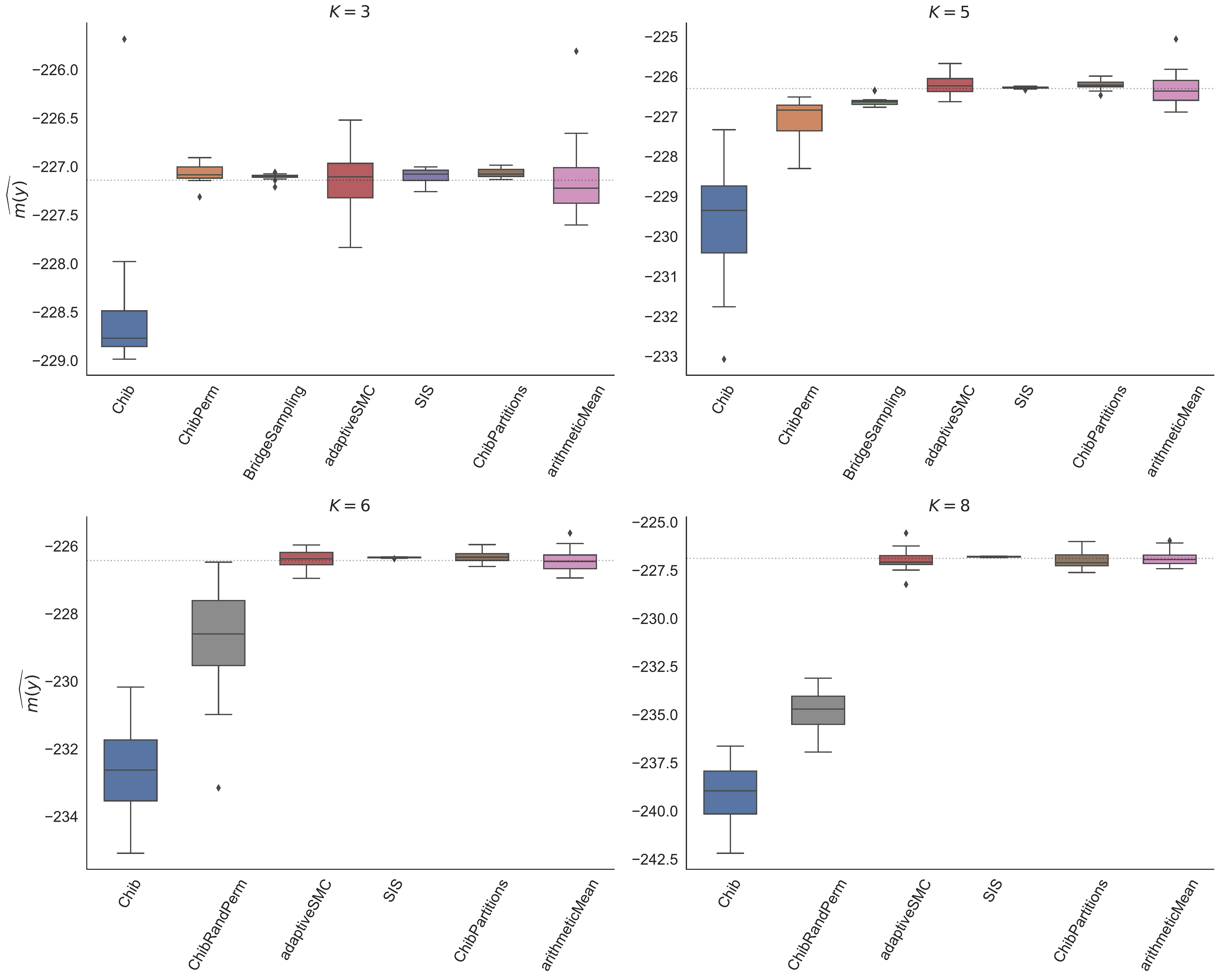}
        \caption{Boxplot of the marginal likelihood estimators 20 repetitions each. Dashed line : mean of the Arithmetic Mean estimator. The tuning parameters are given in Table \ref{tab:hyperParamFM} in Appendix.}
        \label{fig:boxplotFMgalaxy}

    \end{sidewaysfigure}

    For $K=3$, although most estimates agree on a common value for $\widehat{m(x)}$, Chib's method is almost exactly off by a factor $\log 3!\approx 1.79$, which is the consequence of an almost complete lack of label switching during the Gibbs sampling step, as discussed earlier. The sum over all permutations produced by the fully-permuted Chib's estimator (ChibPerm) makes up for this bias. We see that as $K$ grows, all the classical methods except for adaptive SMC fail to estimate the marginal likelihood while our candidates ChibPartition and SIS are consistently pointing to the reference value given by the arithmetic mean estimator.
\begin{table}[]
    \centering
    \begin{tabular}{c|c}
    &Time in seconds\\
    \hline
       ChibPartition  & 425.22 \\
       SIS & 472.87\\
       SMC  & 7539.92\\
       BS & 57983.62\\
    \end{tabular}
    \caption{Average time for $K=5$ required to get one repetition with a precision similar to the boxplot of Figure \ref{fig:boxplotFMgalaxy}. ChibPartition is computed with $T=95000$ iterations after a burn-in of 5000. SIS is iterated 6000 times (i.e $T=6000$). Adaptive SMC has 10000 particles and the \textit{mutation} MCMC is iterated 10 times. For Bridge Sampling, $T_1=1000$ after a burn-in period of 1000, $T_2=1000$ and $T_0=100$}
    \label{tab:tabTimeFMK5}
\end{table}
For $K=5$, Table \ref{tab:tabTimeFMK5} gives the time taken by each of the four successful algorithms (ChibPartition, adaptive SMC, SIS, and bridge sampling) to yield one estimate of the marginal likelihood. The time given corresponds to the average computational time over 20 repetitions of each estimator on one core Intel(R) Xeon(R) CPU E5-2630 v4 @ 2.20GHz. As is expected theoretically, one can observe the great decrease in computational time offered by ChibPartition and SIS with respect to Bridge Sampling and the fully permuted Chib's estimator. Although it is not done in this experiment, note that it is straightforward to parallelize the embarrassingly parallel SIS algorithm and thus to further reduce its computational time. Figure \ref{fig:MSETimeFMgalaxy} shows the evolution of the Mean Squared Error (MSE) as a function of time for the 5 component-mixture model on the \texttt{Galaxy} data. Note that methods SIS and SMC are not implemented in their parallelized version. Despite this, SIS clearly outperforms the other algorithms, closely followed by ChibPartition.

Figure \ref{fig:galaxyNumComp} gives boxplots of the marginal likelihoods estimators given by SIS for different values of $K$ and indicates that a mixture of 5 components is best supported by the $\texttt{galaxy}$ data, for our choice of prior.
\begin{figure}
        
        \centering
        \includegraphics[scale=0.7]{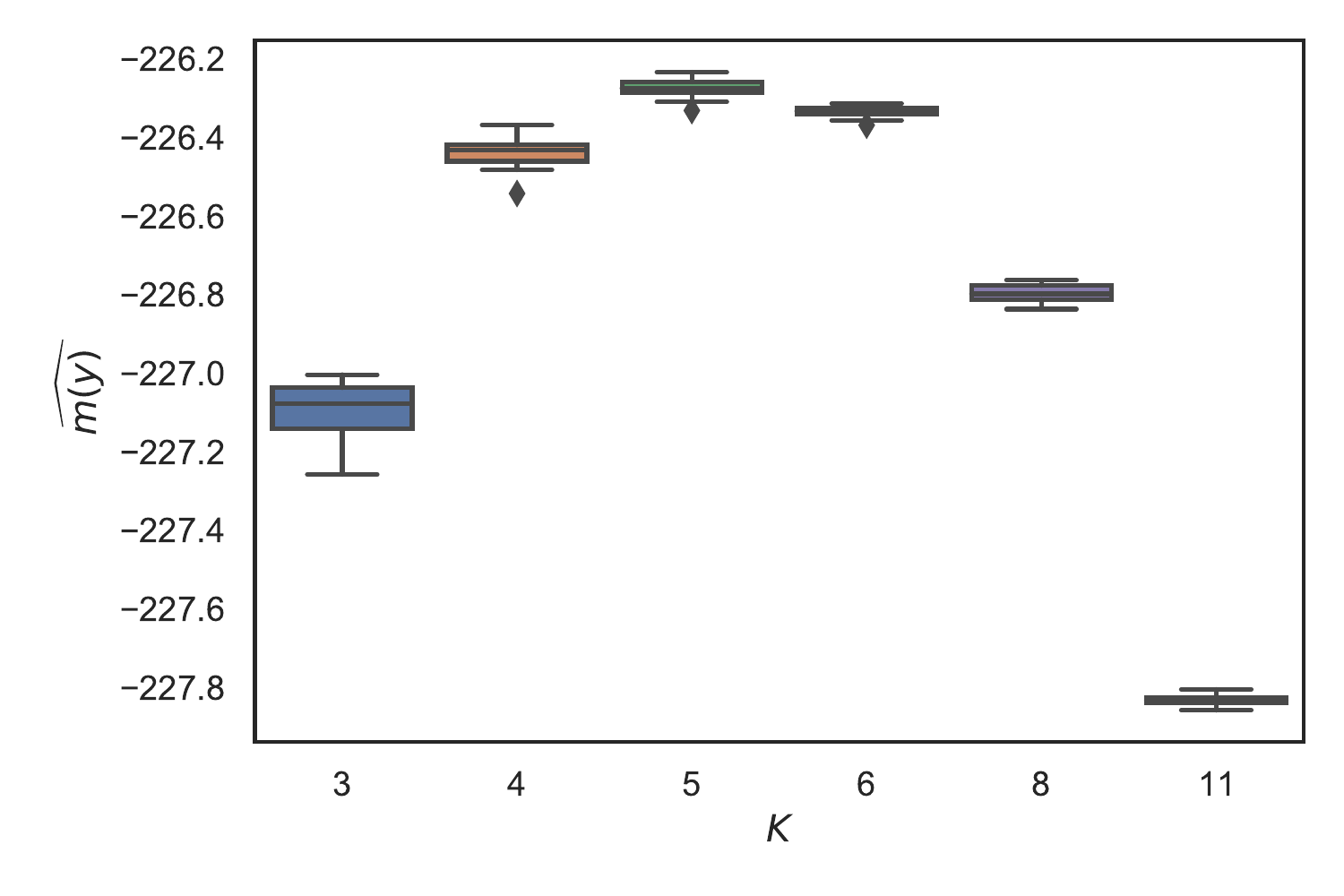}
        \caption{\texttt{Galaxy} data. Boxplots of the  SIS marginal likelihood estimators for different values of $K$.  20 repetitions each.}
        \label{fig:galaxyNumComp}
    \end{figure}

\begin{figure}
        \centering
        \includegraphics[scale=0.46]{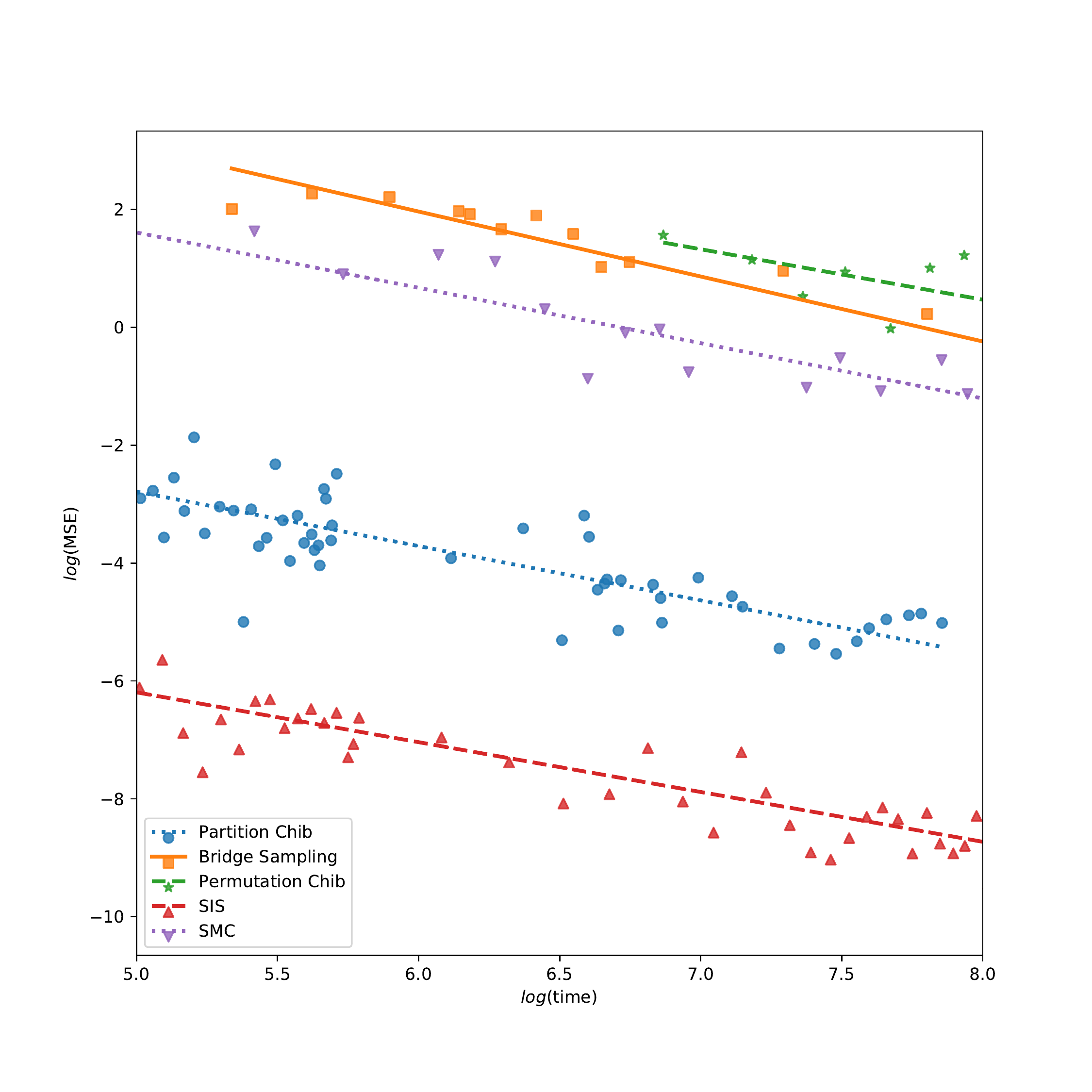}
        \caption{\texttt{Galaxy} data. MSE vs time with regression lines.}
        \label{fig:MSETimeFMgalaxy}
    \end{figure}

\subsubsection{Experiment 2 : Synthetic data, $n=1000$ and $n=2000$.} For more realistic applications, our goal is to find which algorithms scale well as both $K$ and $n$ get large. To our knowledge, no earlier work has been conducted towards identifing reliable estimators for this kind of challenging scenarios. 

We generate two data sets of respectively 1000  and 2000 observations from a six-component mixture of Gaussians. Given the results of the previous experiments in which all but three methods failed to converge in a reasonable time, we only consider ChibPartition, SIS, SMC and Bridge sampling for this more difficult scenario. As in the previous experiment, the conditionally conjugate prior (\ref{condConjPrior}) is chosen.
The tuning parameters used in Figure \ref{fig:boxplotSynthData} and given in Table \ref{tab:hyperParamSynthFM} are chosen so that the running times are comparable.

For the most simple scenario ($n=1000$, $K=3$), ChibPartition is showing a very large variance, probably due to the high cardinality of the set of partitions $\mathcal P_3([1000])$ compared to the number of iterations of the Gibbs sampler. Adaptive SMC is showing good results until $n=2000$ and $K=13$ where its variance becomes very large. On the other hand, SIS is consistently the best estimator in terms of variance, whereas adaptive SMC breaks down in the most demanding experimental setting.\\
The good performance of Bridge sampling should be noted when $n$ is large and $K=3$, although it cannot converge in a decent time when $K=13$ and is hence not displayed in this more difficult scenario.
\begin{figure}[]
        \centering
        \includegraphics[scale=0.28]{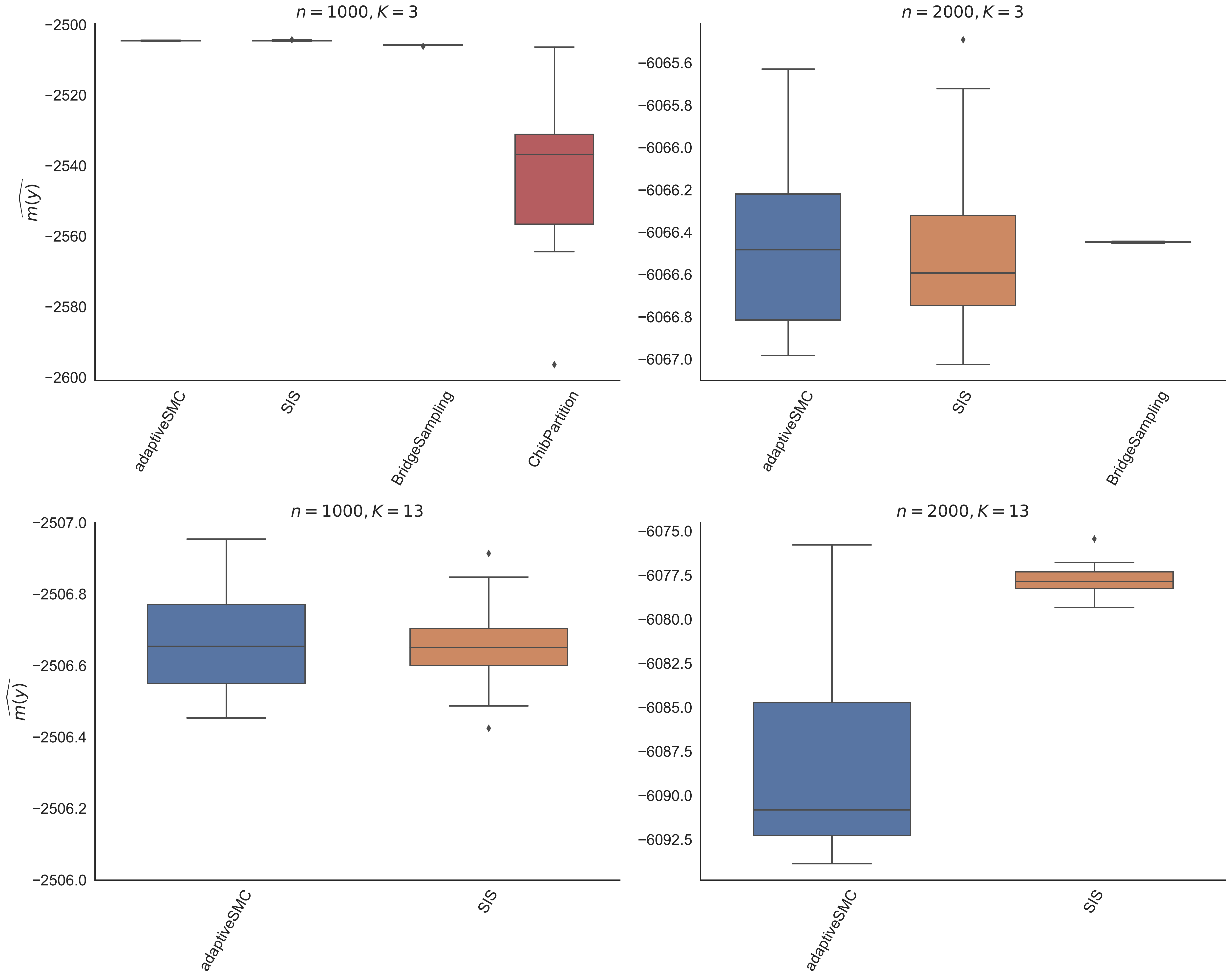}
        \caption{Boxplots for the synthetic data experiment. Tuning parameters are given in Table \ref{tab:hyperParamSynthFM} in appendix}
        \label{fig:boxplotSynthData}

    \end{figure}

\section{Evidence approximation for Dirichlet Process Mixtures}
\label{sec:DPM}
The Dirichlet Process Mixture model arises naturally as the limit when $K\rightarrow\infty$ of the finite mixture model with a Dirichlet prior on the weights $\mathcal D(M/K,\dots,M/K)$ for some $M>0$. The limiting distribution is called the Dirichlet Process Mixture model (DPM) and is characterized by

\begin{align}
y_i|z_i,\boldsymbol \theta &\overset{i.i.d}{\sim} f(y_i|\theta_{z_i}),\;i=1,\dots,n\\
\nonumber p(z_i=k)&=\varpi_k,\;k=1,2,\dots\\
\nonumber\varpi_1,\varpi_2,\dots&\sim GEM(M)\\
\nonumber M&\sim\pi(M)\\
\nonumber\theta_1,\theta_2,\dots&\overset{i.i.d}{\sim} G_0
\end{align}
where $GEM$ denotes the stick-breaking respresentation for the distribution of the weights $(\varpi_k)_k$ as described by \cite{sethuraman1994constructive} and given by
$$\varpi_k=v_k\prod_{i=1}^{k-1}(1-v_i)$$
for $\{v_i\}_i$ an infinite sequence following the $Beta(1,M)$ distribution ensuring that $\sum_{i=1}^\infty \varpi_k=1$ a.s.

Using the above notation, if $P$ is a realisation of the Dirichlet Process with concentration parameter $M$ and base measure $G_0$, denoted $P\sim DP(M,G_0)$, then 
$$P=\sum_{i=1}^\infty \varpi_i\delta_{\theta_i}$$
where $\varpi_1,\varpi_2,\dots$ follow the stick-breaking distribution $GEM$ and $\theta_1,\theta_2,\dots$ are $i.i.d$ realisations from $G_0$.
Hence, the density of a DPM can be written as the infinite mixture
\begin{equation}
    \label{densityDPM}
    f_P(\boldsymbol y)=p(\boldsymbol y|P)=\prod_{i=1}^n\int f(y_i|\theta)dP(\theta)=\prod_{i=1}^n\sum_{i=1}^\infty \varpi_i f(y_i|\theta_i)
\end{equation}

One can derive a useful representation of the DPM model in terms of cluster allocations in the case where the mixture kernel $f$ and the base measure of the Dirichlet Process $G_0$ are conjugate.
Using the Pólya Urn scheme (\cite{blackwell1973ferguson}) one can write
\begin{equation}
\label{spriorcond}
    p(z_i=k|z_1,\dots,z_{k-1},M)=\frac{M}{M+i-1}\delta_{k=\max\limits_{1\leq j\leq i-1}\{z_j\}+1}+\frac{1}{M+i-1}\sum_{j=1}^{i-1}\delta_{k=z_j}
\end{equation}
and thanks to the exchangeability of the sequence $\mathbf{z}=(z_1,\dots,z_n)$ one can derive the induced joint prior distribution on the cluster allocations
\begin{equation}
    \pi(\mathbf{z}|M)=\frac{\Gamma(M)}{\Gamma(M+n)}M^{K_+}\prod_{j=1}^{K_+}\Gamma(n_j)
\end{equation}
where $K_+$ is the number of distinct values among $(z_1,\dots,z_n)$ and for all $j=1,\dots,K_+$, $n_j$ is the number of $z_i$'s equal to $j$.
This representation in terms of cluster allocations also allows for a convenient form of $p(y|z_1,\dots,z_n,M)$ as
\begin{equation}
\label{eq_likelihood}
    p(\mathbf{y}|\mathbf{z},M)=\prod_{j=1}^{K_+}\int\prod_{i:z_i=j}f(y_i|\theta_j)dG_0(\theta_j)
\end{equation}
Obviously this expression of the likelihood function is only useful when an analytic form for (\ref{eq_likelihood}) exists, which is the case when $G_0$ and $f$ are conjugate, which we shall assume hereafter.
Then one can derive the posterior distribution of $(z_1,\dots,z_n,M)$ up to a constant as
\begin{equation}
\label{eq:unnormalizedposterior}
    \pi(\mathbf{z},M|\mathbf{y})\propto p(\mathbf{y}|\mathbf{z},M)\pi(\mathbf{z}|M)\pi(M)
\end{equation}
\cite{neal2000markov} gives an extensive review of Gibbs sampling strategies targeting the posterior of a Dirichlet Process mixture $\eqref{eq:unnormalizedposterior}$.

While literature is abundant on how to derive efficient estimators of the marginal likelihood of a parametric model, as we discussed in the previous section, equivalent methods for non-parametric models such as the DPM model are far from numerous. To our knowledge, only \cite{basu2003marginal} address the issue of evidence computation for the DPM in the conjugate case by extending the well-established method of \cite{chib1995marginal}. It is to be noted that it is difficult to find examples of application in literature of this non-parametric version of Chib's algorithm. There exists an SMC framework summarized by \cite{griffin2011sequential} that is applicable if $M$ is assumed to be known. This is not our approach in this article since, as illustrated on Figure 3 of \cite{tokdar2021bayesian}, this parameter can have a decisive influence on the marginal likelihood and thus on the Bayes Factor. 

In this section we aim at assessing the method proposed in \cite{basu2003marginal} and deriving alternative ways to compute the marginal likelihood of a DPM. 

\subsection{Existing algorithm}
\paragraph{Chib's algorithm.}
\cite{basu2003marginal} adapt Chib's algorithm from \cite{chib1995marginal} to the conjugate Dirichlet Process Mixture model using Bayes identity
\begin{equation}
    \label{ChibDPM}
    \hat m_{DP}^{Chib}(\boldsymbol y)=\frac{ L(\boldsymbol y|M^*,G_0)\pi(M^*)}{\hat \pi(M^*|\boldsymbol y)}
\end{equation}
where $M^*$ is some point in $(0,\infty)$ and $L(\boldsymbol y|M^*,G_0)$ is the integrated likelihood with respect to the Dirichlet Process $DP(M,G_0)$,
\begin{equation}
    \label{integrated_llk_DPM}
    L(\boldsymbol y|M,G_0)=\int f_P(\boldsymbol y)dDP(P|M,G_0)
\end{equation}
 Just like for its finite counterpart, Chib's algorithm for the DPM uses a Rao-blackwellised estimator of the posterior density with the introduction of a latent variable $\eta$ in the Gibbs sampler, as suggested by \cite{escobar1995bayesian},
\begin{equation}
\label{eqMstar}
    \hat\pi(M^*|\boldsymbol y)=\sum_{t=1}^T\pi(M^*|\boldsymbol y, \eta^{(t)},K_+^{(t)})
\end{equation}
where $K_+$ is the number of non-empty clusters implied by the allocation vector $z$.
The full conditional distribution on $M$ is available in closed form provided $M$ follows a Gamma distribution $\Gamma(a,b)$ a priori and is given by the mixture
\begin{equation*}
    M|\boldsymbol y,\eta,K_+\sim\omega\Gamma(a+K_+,b-\log(\eta))+(1-\omega)\Gamma(a+K_+-1,b-\log(\eta))
\end{equation*}
where $\omega=(a+K_+-1)/\{n(b-log(\eta))+a+K_+-1\}$ and $\eta|\boldsymbol y,M\sim Beta(M+1,n)$.

Unlike the finite mixture model, the likelihood ordinate (\ref{integrated_llk_DPM}) is defined by an intractable integral which must be estimated. Following \cite{kong1994sequential}, the authors propose a Sequential Importance Sampling (SIS) scheme using the following importance distribution where for any vector $\boldsymbol x\in \mathbb R^n$ we define $\boldsymbol x_{1:i}:=(x_1,\dots,x_i)$, for all $i=1,\dots,n$. 

\begin{equation}
    \label{newton}
    \pi^*(z_1,\dots,z_n|\boldsymbol y,M)=\prod_{i=1}^n\pi(z_i|\boldsymbol y_{1:i},\boldsymbol z_{1:i-1},M)
\end{equation}
which is available in closed form in the conjugate case. It can be easily derived that the final importance weight can be expressed as 
\begin{equation*}
    w_n(z_1,\dots,z_n)=\frac{\pi(z_1,\dots,z_n|\boldsymbol y,M)}{\pi^*(z_1,\dots,z_n|\boldsymbol y,M)}=\frac{p(y_1|z_1,G_0)\prod_{i=2}^n p(y_i|\boldsymbol y_{1:i-1}\boldsymbol z_{1:i-1},G_0)}{L(\boldsymbol y|M^*,G_0)}
\end{equation*}
Hence, given a sequence $\{\boldsymbol z^{(t)}\}_{t=1}^T$ generated sequentially from (\ref{newton}), one can derive as a by-product of SIS 
\begin{equation}
    \hat L(\boldsymbol y|M^*,G_0)=\frac{1}{T}\sum_{t=1}^T p(y_1|z_1^{(t)},G_0)\prod_{i=2}^n p(y_i|\boldsymbol y_{1:i-1}\boldsymbol z^{(t)}_{1:i-1},G_0)
\end{equation}
Algorithm \ref{algoChibDPM} gives the details of implementation of this variant of Chib's algorithm.
\begin{algorithm}
\setstretch{1.35}
\caption{\textit{Chib's algorithm for the DPM}}\label{algoChibDPM}
\textbf{Input} : $(\boldsymbol z^{(t)})_{t=1}^{T_1},( K_+^{(t)})_{t=1}^{T_1},( \eta^{(t)})_{t=1}^{T_1}$ from a Markov Chain at stationarity targeting $\pi_K(\boldsymbol z,M|\boldsymbol y)$ \\
\For {$t=1,\dots,T_1$}{
    $\text{Compute } \hat\pi(M^*|\boldsymbol y)= \sum_{t=1}^{T_1}\pi(M^*|\boldsymbol y,\eta^{(t)},K_+^{(t)})$ using (\ref{eqMstar})
}

\For {$q=1,\dots,T_2$}{
    Set $\tilde z_1^{(q)}$=1, $K_+=1$ and compute $\gamma_1=p(y_1|\tilde z_1,G_0)=\int f(y_1|\theta)dG_0(\theta)$\\
    \For{$i=1,\dots,n$}{
    Compute \begin{align*}
        \gamma_i^{(q)}&=p(y_i|\boldsymbol y_{1:i-1},\boldsymbol{\tilde z^{(q)}_{1:i-1}},M^*)\\
        &=\frac{M^*}{M^*+i-1}\int f(y_i|\theta)dG_0(\theta)+\sum_{k=1}^{K_+}\frac{n_k(\boldsymbol{\tilde z}^{(q)}_{1:i-1})}{M^*+i-1}\frac{m(C_k(\boldsymbol{\tilde z}^{(q)}_{1:i-1})\cup\{y_i\})}{m(C_k(\boldsymbol{\tilde z}^{(q)}_{1:i-1}))}
    \end{align*}
    Draw $\boldsymbol{\tilde{ z}}^{(q)}_i$ from the categorical distribution
    \begin{align*}
        p(\boldsymbol{\tilde{z}}=k|\boldsymbol y_{1:i},\boldsymbol{\tilde{z}}^{(q)}_{1:i-1}, M^*)\propto\begin{cases}
        \frac{n_k(\boldsymbol{\tilde{z}}^{(q)})}{M^*+i-1}\frac{m(C_k(\boldsymbol{\tilde{z}}_{1:i-1}^{(q)})\cup\{y_i\})}{m(C_k(\boldsymbol{\tilde{z}}_{1:i-1}^{(q)}))} \text{ if } k=1,\dots,K_+\\
        \frac{M^*}{M^*+i-1}\int f(y_i|\theta)dG_0(\theta)\text{ if } k=K_++1
        \end{cases}
    \end{align*}
    }
    Compute $w^{(q)}=\prod_{i=1}^n\gamma_i^{(q)}$
}
Compute $\hat L(\boldsymbol y|M^*,G_0)=1/T_2\sum_{q=1}^{T_2}w^{(q)}$

\textbf{Output} : $m_{DP}^{Chib}(\boldsymbol y)=\hat L(\boldsymbol y|M^*,G_0)\pi(M^*)/\hat\pi(M^*|\boldsymbol y)$
\end{algorithm}

\subsection{Proposed algorithm}
The adapted version of Chib's algorithm proposed by \cite{basu2003marginal} is very efficient provided a sensible choice is made for the value of $M^*$ in equation \eqref{eqMstar}. However, it is difficult to identify a candidate value with high posterior probability since $p(\boldsymbol y|M)$ is difficult to compute for all $M$. We here propose another approach related to Bridge Sampling that we believe to be more stable.
\paragraph{Reverse Logistic Regression.} Introduced by \cite{geyer1994estimating}, Reverse Logistic Regression (RLR) is a biased IS-related method. It requires $T_1$ and $T_2$ $i.i.d$ observations from two distributions $\pi_1=\tilde{\pi}_1/c_1$, the adversarial distribution for which the normalizing constant $c_1$ is known, and $\pi_2=\tilde{\pi}_2/c_2$, the distribution of interest, only known up to a normalizing constant $c_2$. The main idea of Reverse Logistic Regression is to ignore from which distribution each observation stems (i.e ‘forgetting' the labels). By doing so, observations can be assumed to arise from a mixture model which mixture weights are proportional to the normalizing constants $c_1$ and $c_2$. The idea of \cite{geyer1994estimating} is then to perform some ‘reverse' logistic regression, in the sense that the response variable is not random while the predictors are, in order to retrieve an estimator for $c_2$.

We here apply this idea to the DPM by considering as an importance distribution $\pi_1(\boldsymbol z,M):=\pi^*(\boldsymbol z|\boldsymbol y,M)\pi(M)$, where $\pi^*$ is defined as in (\ref{newton}), and $\pi_2(\boldsymbol z,M)=\pi_2(\boldsymbol z,M)/c_2:=\tilde \pi(z,M|\boldsymbol y)/m(\boldsymbol y)$ is the posterior distribution. Note that simpler choices for $\pi_1$ exists, in particular for the non-conjugate case, such as the prior distribution $\pi(\boldsymbol z,M)$ for instance.\\
Assume that $\{\boldsymbol z^{(1,j)},M^{(1,j)}\}_{j=1}^{T_1}$ and $\{\boldsymbol z^{(2,j)},M^{(2,j)}\}_{j=1}^{T_2}$ are respectively i.i.d samples from $\pi_1$ and $\pi_2$.
Then the marginal likelihood $c_2:=\hat m(\boldsymbol y)$ can be estimated by finding the maximum of the log quasi-likelihood
\begin{align}
\label{eq:objectiveGeyer}
    \ell(c_2)&=\sum_{i=1}^{T_1} \log\gamma_1(\boldsymbol z^{(1,i)},,M^{(1,i)})+
    \sum_{j=1}^{T_2} \log\gamma_2(\boldsymbol z^{(2,j)},,M^{(2,j)})
\end{align}
where
    $$\gamma_1(\boldsymbol z,M)=\frac{\frac{T_1}{T_1+T_2}\pi_1(\boldsymbol z,M)}{\frac{T_1}{T_1+T_2}\pi_1(\boldsymbol z,M)+\frac{T_2}{T_1+T_2}\tilde \pi_2(\boldsymbol z,M)/c_2}$$
and
    $$\gamma_2(\boldsymbol z,M)=1-\gamma_1(\boldsymbol z,M)$$
 \cite{chen1997monte} show that Reverse Logistic Regression is essentially equivalent to optimal Bridge sampling. Therefore, standard results about Bridge Sampling as given in \cite{fruhwirth2004estimating} can be applied to derive the standard error of the simulated marginal likelihood $\hat m_{DP}^{RLR}(\boldsymbol y)$.
 
   \begin{algorithm}[!ht]
   \setstretch{1.7}
\caption{\textit{RLR algorithm for the DPM - SIS instrumental distribution}}\label{}
\textbf{Input} : $\{\boldsymbol z^{(1,t)}\}_{t=1}^{T_1},\{M^{(1,t)}\}_{t=1}^{T_1}\sim\pi^*(\boldsymbol z|\boldsymbol y,M)\pi(M)$  and $\{\boldsymbol z^{(2,t)}\}_{t=1}^{T_2},\{M^{(2,t)}\}_{t=1}^{T_2}$ from a Markov Chain at stationarity targeting $\pi_K(\boldsymbol z,M|\boldsymbol y)$ \\
\For {$t=1,\dots,T_1$}{
    $\text{Compute } \pi_1(\boldsymbol z^{(1,t)},M^{(1,t)})=\pi^*(\boldsymbol z^{(1,t)}|M^{(1,t)},\boldsymbol y)\pi(M^{(1,t)})$ (\textit{this step can be avoided if done alongside the simulation of} $\{\boldsymbol z^{(1,t)}\}_{t=1}^{T_1},\{M^{(1,t)}\}_{t=1}^{T_1}$)\\
    Compute $\tilde\pi_2(\boldsymbol z^{(1,t)},M^{(1,t)})=p(\boldsymbol y|\boldsymbol z^{(1,t)},M^{(1,t)})\pi(\boldsymbol z^{(1,t)}|M^{(1,t)})\pi(M^{(1,t)})$ as in (\ref{eq:unnormalizedposterior})
}

\For {$t=1,\dots,T_2$}{
    $\text{Compute } \pi_1(\boldsymbol z^{(2,t)},M^{(2,t)})=\pi^*(\boldsymbol z^{(2,t)}|M^{(2,t)},\boldsymbol y)\pi(M^{(2,t)})$\\
    Compute $\tilde\pi_2(\boldsymbol z^{(2,t)},M^{(2,t)})=p(\boldsymbol y|\boldsymbol z^{(2,t)},M^{(2,t)})\pi(\boldsymbol z^{(2,t)}|M^{(2,t)})\pi(M^{(2,t)})$ as in (\ref{eq:unnormalizedposterior})
}

Optimize function $\ell(c_2)$ defined in (\ref{eq:objectiveGeyer})\\
\textbf{Output} : $\hat m_{DP}^{RLR}(\boldsymbol y)=\arg\max\ell(c_2)$
\end{algorithm}

\subsection{Simulation study}

\subsubsection{Experiment 3 : \texttt{Galaxy} data.}
 In this section, we evaluate the relative performances of \cite{basu2003marginal}'s algorithm with other competing alternatives which, to our knowledge, has never been done. To do so, we once again start by using the \texttt{galaxy} data and consider a Dirichlet Process mixture of normal distributions with unknown location and scale parameters and use the conditionally conjugate Normal-Inverse Gamma prior $G_0$ as defined in (\ref{condConjPrior}).\\
 Note that, while,  for finite mixtures, model complexity is a function of the number of mixture components $K$, most difficulties arise in the DPM for growing numbers of observations $n$. This is partly due to the fact that 
 %$K_+\overset{+\infty}{\sim}M\log(n)$. 
 $K_+=\text{O}(M\log(n))$.
 Hence we design three different experiments with respectively 6, 36 and 82 observations from the $\texttt{Galaxy}$ data as shown on Figure \ref{fig:galaxyDPM}. One can note in particular that simple estimators such as the Arithmetic Mean and Harmonic Mean estimators, which show rather good results for a small data size, fail to converge as soon as the amount of data becomes moderately large. Since we are not able to make the Arithmetic Mean estimator converge (i.e to make its variance low) as in the finite mixture experiments of Section \ref{secsimulationsFM}, we use Reverse Logistic Regression where the prior is used as the adversarial distribution (hereafter RLR-Prior) as a reference value. This method is indeed independent of both \cite{basu2003marginal}'s algorithm and RLR with a SIS adversarial distribution since it does not depend on the sequential imputation scheme \eqref{newton}. It is interesting to note that RLR-Prior yields rather satisfactory results given the simplicity of the adversarial distribution, compared with the Arithmetic Mean estimator that also relies on prior samples and suffers from pathological variance. For non-conjugate Dirichlet Process Mixture models, this rather inexpensive estimator could be a good, easy-to-implement first approach to the marginal likelihood estimation problem.

 \begin{figure}
     
        \centering
        \includegraphics[scale=0.24]{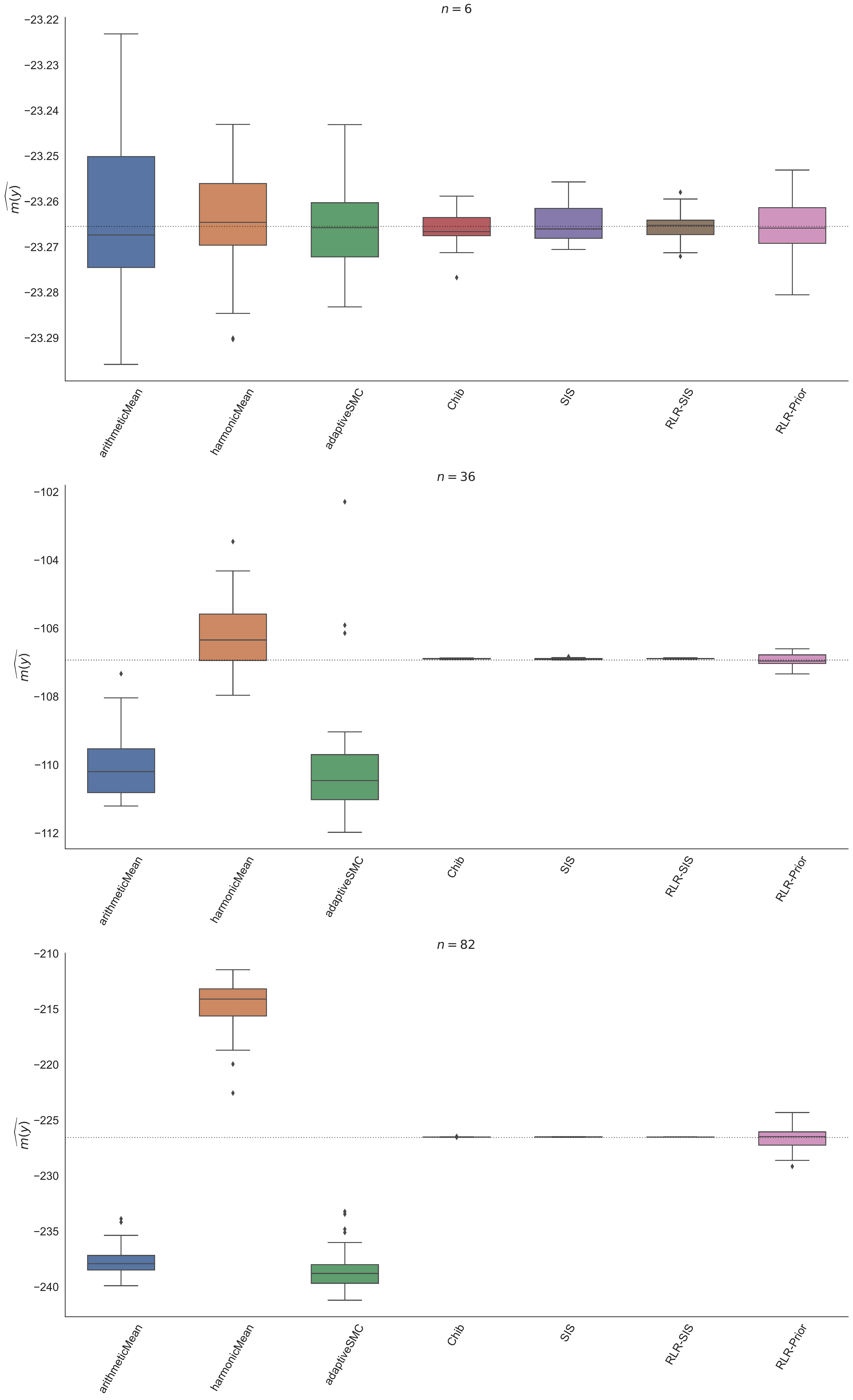}
        \caption{\texttt{Galaxy} data. Boxplots of the different methods (20 repetitions each) for different values of $n$. The dashed line corresponds to the mean of the RLR-prior estimator. The choice of tuning parameters is given in Table \ref{tab:hyperParamDPM} in Appendix.}
           \label{fig:galaxyDPM}

 \end{figure}

Figure \ref{fig:MSEvsTimeDPM} shows the decrease of the MSE as a function of time and shows that for a given allocated time, the error produced by Chib's estimator is greater than that of RLR-SIS by about a factor $\exp(0.3)\approx 1.35$.
\begin{figure}[]
        \centering
        \includegraphics[scale=0.46]{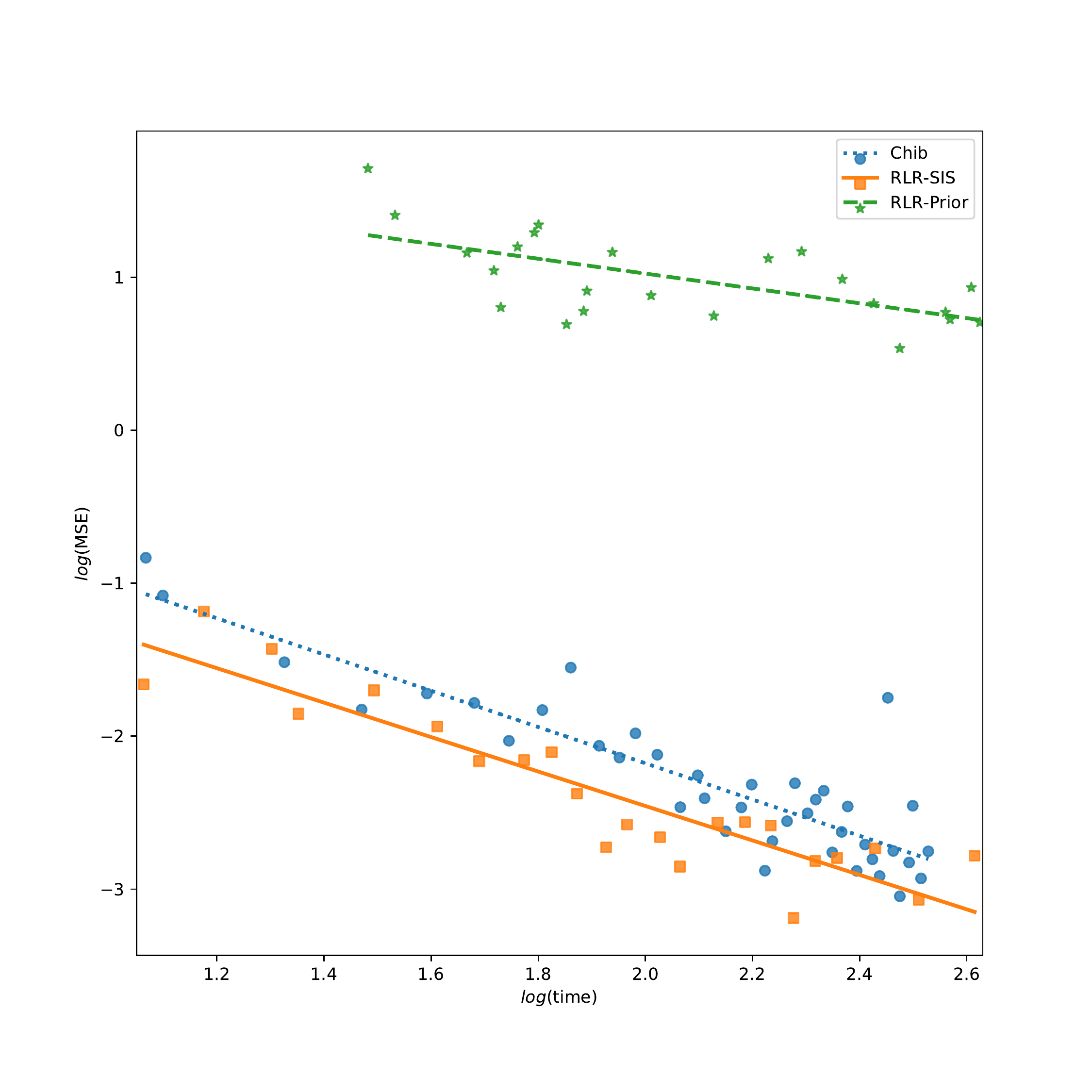}
        \caption{\texttt{Galaxy} data. MSE vs time with regression lines.}
        \label{fig:MSEvsTimeDPM}

    \end{figure}
%\section{Asymptotic properties of the marginal likelihood of a DPM}

\subsubsection{Experiment 4 : Synthetic data, $n=1000$.}
We now assess the scalability of both methods on a synthetic data set of 1000 observations, arising from a 6-component mixture of normal distributions. Although we cannot compute analytically the true value of the marginal likelihood, we do know that both Chib's algorithm and RLR-SIS are consistent, as illustrated in the previous section. That is, provided we do not observe a large simulation-related variance, the estimator can be trusted. The tuning parameters used for the experiment, as well as the run time are given in the caption of Figure \ref{fig:n1000DPM}. The run time corresponds to the total time required to compute 20 repetitions of each estimator on 10 cores, each core being an Intel(R) Xeon(R) CPU E5-2630 v4 @ 2.20GHz. Note that no parallelization of the embarrassingly parallel SIS step was done (which would decrease the total run time of both algorithms dramatically, but by the same factor).

On Figure \ref{fig:n1000DPM}, it is clear that Chib's estimator suffers from a large variance, which in turn translates into a downward bias on the $log$-scale. On the other hand, RLR-SIS exhibits much more stability. Note that although more iterations and computational time are allocated to Chib's estimator, this does not seem to be enough to correct Chib's pathological variance while RLR-SIS yields a more accurate estimate of the marginal likelihood within a much lower run time.

\begin{figure}
    \centering
    \includegraphics[scale=0.5]{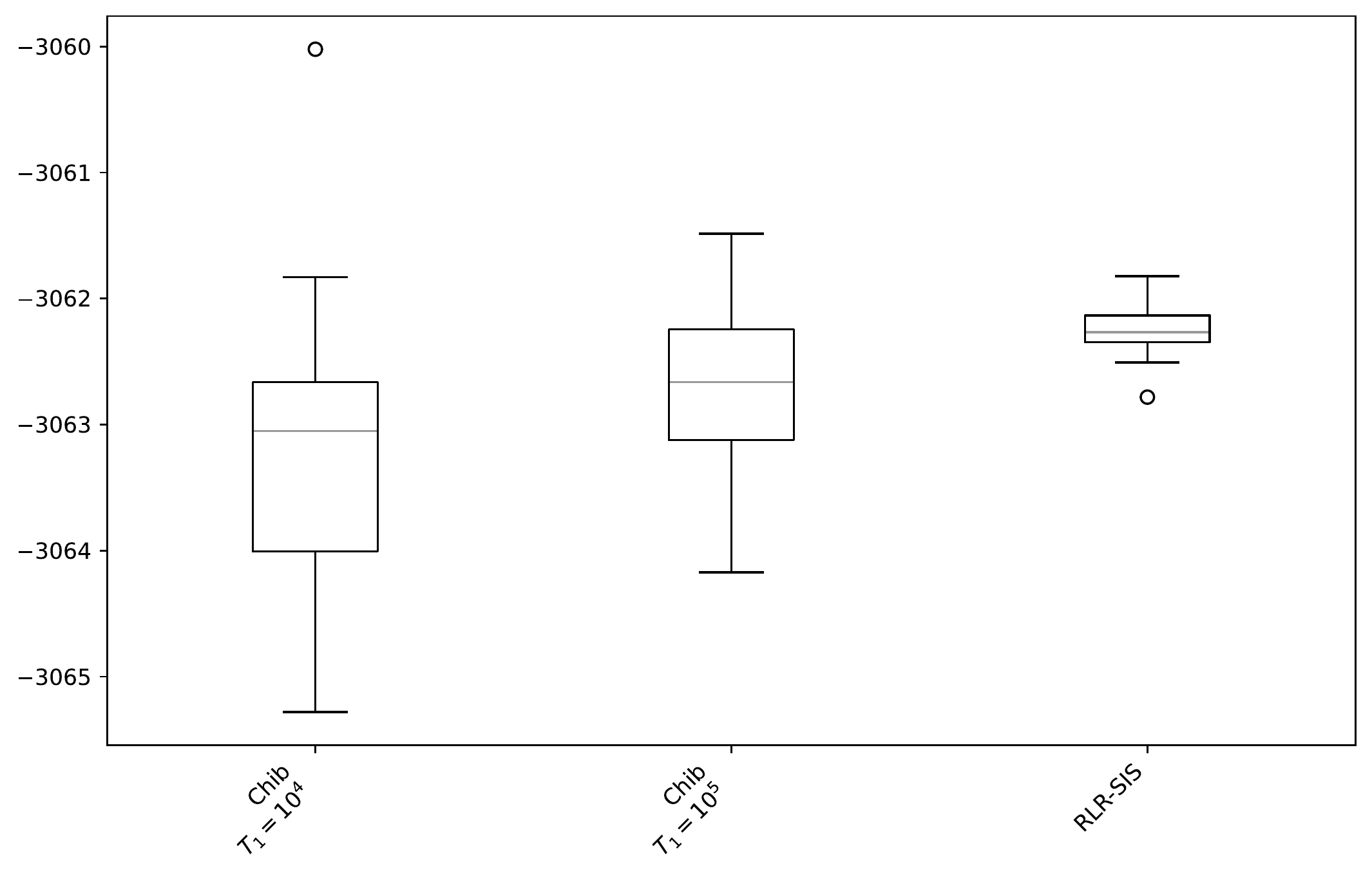}
    \caption{Boxplots of the marginal likelihoods estimates, 20 repetitions each. Dirichlet process mixture model, synthetic data, $n=1000$. Chib left hand side: Chib : $T_1=10^4$, $burnIn=10^3$, $T_2=600$, Run time : 04:09:51. Chib right hand side $T_1=10^5$, $burnIn=10^4$, $T_2=2000$, Run time : 34:01:18. RLR-SIS : $T_1=10^4$, $burnIn=10^3$, $T_2=600$, Run time : 06:23:12.}
    \label{fig:n1000DPM}
\end{figure}

\subsubsection{Experiment 5 : Testing a finite mixture against a DPM.} 
In this experiment, we test whether the Bayes factor converges to infinity under the null hypothesis that data arises from a finite mixture with $K_0$ components, when fitting a DPM. That is, we want to check whether $$ BF:=\frac{m_{K_0}(\boldsymbol y)}{m_{DPM}(\boldsymbol y)}\underset{n\rightarrow\infty}{\longrightarrow}\infty$$ To do so, 100 data sets $\boldsymbol y$ are generated from $P_{K_0}$, where $P_{K_0}$ is a finite mixture of $K_0$ normal components. The Bayes Factor is then computed for successive values of $n$ and the resulting ‘Bayes Factor paths' are displayed on Figure \ref{fig:BFDPMFM}. On the left hand side, $K_0=1$ and the datasets simply arise from a normal distribution with mean $\mu=0$ and scale $\sigma=2$. On the right hand side, $K_0=3$ with mean parameters $\boldsymbol \mu=(-3,4,12)$, scales $\boldsymbol \sigma=(2,2,2)$ and weights $\boldsymbol \varpi=(0.3,0.2,0.5)$. Empirically, we can see that the Bayes Factor appears to converge to infinity, in both scenarios as it seams that $P_{K_0}(BF>1)\longrightarrow 1$. 

In the next section we formalize this empirical study by proving the consistence of the Bayes Factor for a particular class of ‘strongly identifiable' Dirichlet Process Mixture models under a finite mixture null hypothesis.
  \begin{figure}[]
        \centering
        \includegraphics[width=\textwidth]{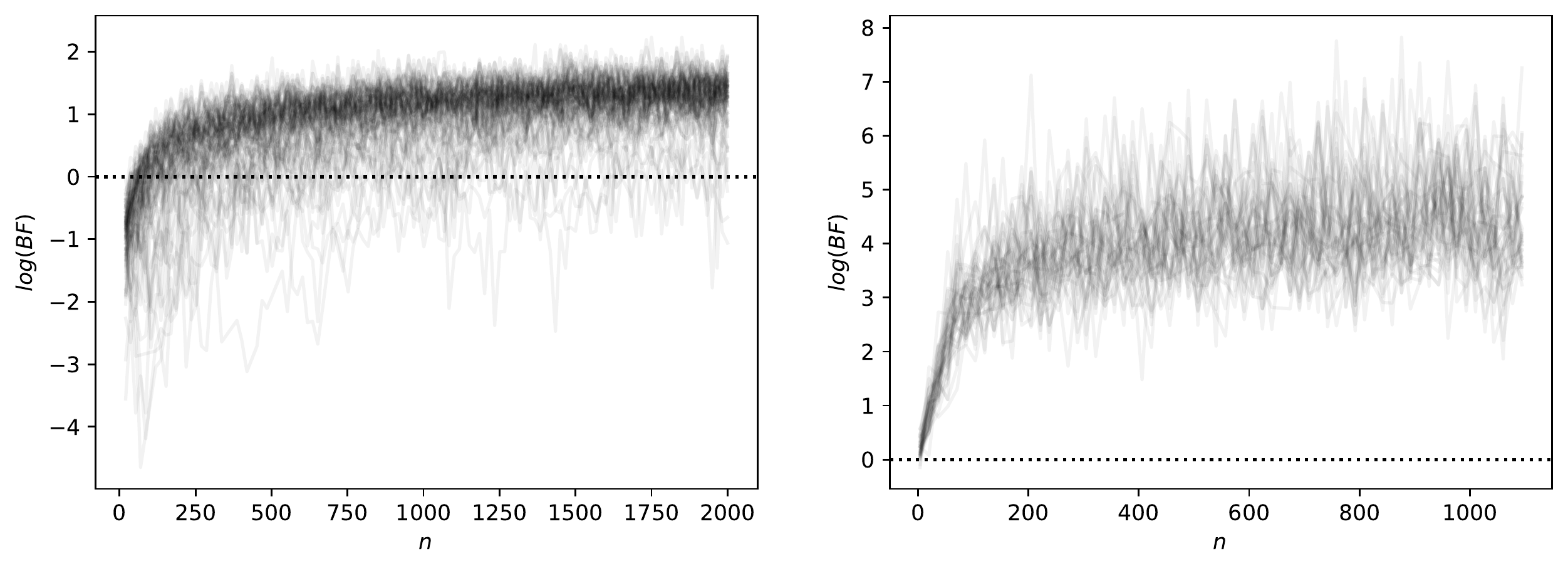}
        \caption{$log$-Bayes Factor paths for 100 datasets as $n$ increases. Left-hand side : normal distribution. Right-hand side : 3-component mixture of normal distributions. The dashed line is the line $\log(BF)=0$.}
        \label{fig:BFDPMFM}

    \end{figure}
\section{Asymptotic behaviour of the marginal likelihood associated to the Dirichlet process mixture} \label{sec:theory}

In this section we study the asymptotic behaviour of the marginal likelihood 
$m_{DP}(\mathbf y) = \int f_{P}(\mathbf y) d\pi(P)$ when $\pi $ is a $DP(M, G_0)$ and $\mathbf y = (y_1, \cdots, y_n) $ with $y_i \stackrel{iid}{\sim} f_0$, a probability density on $ \mathbb R^{d}$.

Understanding the behaviour of $m_{DP}(\mathbf y)$ corresponds to determining asymptotic lower and upper bounds for $m_{DP}(\mathbf y)$. Determining lower bounds on marginal densities is typically done using the technique of \cite{ghosal2000convergence}, Lemma 8.1 for instance,  where it is  is used to derive posterior concentration rates under a Dirichlet Process mixture model.  
There is now a large literature on  posterior contraction rates in Dirichlet process mixture models, see for instance \cite{ghosal2007posterior}, \cite{kruijer2010adaptive}, \cite{shen2013adaptive}, and \cite{scricciolo2014adaptive} in which a lower bound on $m_{DP}(\mathbf y)$ is derived. 

The difficult part when assessing evidence in this setting stands in obtaining an upper bound on $m_{DP}(\mathbf y)$, since it requires a refined understanding on neighbourhoods of $f_0$. In the following  section we concentrate on  deriving such an upper bound when $f_0 \in \cup_{K\in \mathbb N^\star} \model_K$, where $\model_K$ denotes the  model of mixtures with  $K$  densities $f_\theta(y)$. Indeed an important application of such an upper bound is in the goodness of fit test (or test for the number of components)
$$ H_0 :\,  f_0 \in \model_K, \quad \text{versus} \quad H_1: \,f_0 \notin  \model_K,$$
to prove that the Bayes factor is consistent under the null:
$$ \frac{ m_K(\mathbf y) }{ m_{DP}(\mathbf y) } \rightarrow \infty , \quad \text{in probability under} f_0 $$

  Obtaining such an upper bound is of interest even outside the context of testing, since it is a way to understand the behaviour of credible regions in infinite dimensional models (see, e.g., \cite{rousseau2020asymptotic}), together with proving a lower bound on posterior contraction rates as in \cite{castillo2008lower}. 

\subsection{Upper bound on $m_{DP}(\mathbf y)$} \label{sec:theorie:UB}

In this section we assume that $P_0 = \sum_{j=1}^{k_0}\varpi_j^0\delta_{\theta_j^0}$ where $\varpi_j^0>0$ for all $j$ and $\theta_j^0\neq \theta_i^0$ for all $i\neq j$.

We also consider the following regularity assumptions on the distribution $f(y|\theta)$. 

\noindent
\textbf{Assumption A1} [Regularity] For all $y\in \mathcal Y\subset \mathbb R^{d}$, the function $\theta \rightarrow f(y|\theta) $ is twice continuously differentiable and there exist $H_1\in L^2(\mathbb R^d),\,H_2 \in L^1(\mathbb R^d),\, H_3 \in L^1(\mathbb R^d)$, and $\delta_0,\delta>0$ such that for all $j=1, \cdots, k_0$, 
\begin{equation}
\begin{split}
 \sup_{\theta \in \Theta}  \| \nabla f(y| \theta) \| &\leq H_1(y) , \quad 
  \sup_{\|\theta - \theta_j^0\|\leq \delta_0} \| D^2 f(y| \theta) \|\leq H_2(y)\\
  |D^2 f(y|\theta) - D^2f(y|\theta_j^0)| &\leq H_3(y) \|\theta - \theta_j^0\|^\delta, \quad \forall \|\theta - \theta_0\|\leq \delta_0,
  \end{split}
  \end{equation}

We then consider two types of mixture models: location mixtures and strongly identifiable mixtures. The latter are defined by the following assumption: Denote by $S_d^+$ the set of symmetrical semi-definite matrices of dimension $d$, \\
\textbf{Assumption A2} [Strong identifiability]\\
For all  $\epsilon>0$ and all 
$\nu_0$  non null measure on $A_0 = [\cup_{j=1}^{k_0} B(\theta_j^0, \epsilon)]^c$ satisfying $\nu_0(A_0)\leq 1$ and all $\alpha_0\geq 0$, $\alpha_j \in \mathbb R$, $\beta_j \in \mathbb R^d$ and $\gamma_j \in S_d^+$, $j=1, \cdots, k_0$, 
$$ \alpha_0 f_{\nu_0}(y) + \sum_{j=1}^{k_0} [\alpha_jf_{\theta_j^0} + \beta_j \nabla f_{\theta_j^0} + \gamma_j^t D^2 f_{\theta_j^0} \gamma_j ] = 0 \quad \Leftrightarrow \alpha_0= \alpha_j = 0, \quad \beta_j=\gamma_j=0, \, \forall j$$
%\textbf{Assumption A3} [Compactness ]
%$$ \int \sup_{\theta \in \Theta} f_\theta(x) d\mu(x) < \infty $$
\textbf{Assumption A3} 
For all $x$, $f_\theta(x)$ goes to 0 at  $\bar\Theta\cap \Theta^c$, where $\bar \Theta$ is the closure of $\Theta$ and
$$\sup_\theta \|f_\theta(\cdot)\|_\infty< \infty$$

The cornerstone of the proof of Theorem \ref{UBMDP1} is the existence and boundedness of the density function of the Dirichlet Process random mean defined as $\mu(P):=\int\theta dP(\theta)$ for $P$ distributed from $DP(M,G_0)$. \cite{feigin1989linear} show that the following assumption is a necessary and sufficient condition for the mean $\mu(P)$ to exist.

\noindent\textbf{Assumption A4}
The Dirichlet Process base measure $G_0$ is such that
$$\int_\Theta \log(1+\|\boldsymbol\theta\|_\infty)dG_0(\boldsymbol\theta)<\infty$$

\begin{remark}
A sufficient condition for Assumption A4 to be true is that distribution $G_0$ has a finite first moment, which includes Gaussian distributions for instance. Note that although the Cauchy distribution has no finite expectation, it does verify Assumption A4. 
\end{remark}
\begin{remark} Let $\varphi$ denote the characteristic function of $G_0$.
Since for all $(\varpi_1,\varpi_2,\dots)$ following the stick-breaking distribution and for all $z\in\mathbb R^d$, $$\prod_{j=1}^\infty\varphi(\varpi_j\boldsymbol z)=\mathbb E_\theta\left[\exp\left(i\boldsymbol z^T\sum_{j=1}^\infty \varpi_j\boldsymbol \theta_j\right)\right]=\mathbb E_\theta\left[\exp\left(i\boldsymbol z^T\mu(P)\right)\right]$$
then, $\prod_{j=1}^\infty\varphi(\varpi_j\boldsymbol z)$ exists almost surely if $\mu(P)$ exists almost surely. Hence, Assumption A4 implies the existence of $\prod_{j=1}^\infty\varphi(\varpi_j\boldsymbol z)$.
\end{remark}

The next Theorem shows that under Assumptions \textbf{A1}, \textbf{A2}, \textbf{A3} and \textbf{A4} and if $f_{P_0}$ is a mixture with $k_0$ components, then $m_{DP}(\boldsymbol y)$ is bounded from above by $o( n^{-(k_0-1 + dk_0+t)/2} )$ for some $t>0$:
\begin{theorem}\label{UBMDP1}
Assume that $y_1, \cdots, y_n$ are iid $f_{P_0}$ with $P_0 = \sum_{j=1}^{k_0}\varpi_{j}^0 \delta_{\theta_j^0}$ such that Assumptions \textbf{A1, A2, A3, A4} are satisfied. 
Consider a Dirichlet process prior on $P$: $P \sim DP(M, G_0)$ with $G_0$ satisfying $G_0(\|\theta\|>u) \leq e^{-a_0u^\tau}$ for some $a_0, \tau>0$ when $u$ is large enough. 
Then there exists $t>0$ such that for all $\epsilon>0$
\begin{equation} \label{marginalLB}
\mathbb P_{f_{0}}\left( m_{DP}(\boldsymbol y ) > n^{-(k_0-1 + dk_0+t)/2}\right) =o(1)
\end{equation}
Moreover there exists $q\geq 0$ such that 
\begin{equation}\label{L1concentration}
\Pi_{DP} \left( \| f_0 - f_p\|_1 \leq \frac{(\log n)^q }{ \sqrt{n} } | \boldsymbol y \right) = 1 + o_{P_{f_0}}(1).
\end{equation}
\end{theorem}

\begin{remark}
Exploiting the proof of Proposition 3.10 of \cite{gassiat2014local} (see Proposition 3.12) Assumption \textbf{A2} is established in location mixtures, i.e. for models in the form 
$$f_\theta(y) = f(y-\theta)$$
 and Theorem \ref{UBMDP1} is valid under Assumptions \textbf{A1, A3, A4}. 
\end{remark}

\begin{remark}
The posterior concentration in $L_1$ distance given in \eqref{L1concentration} has been obtained for Dirichlet Process mixtures when the kernel $f_\theta(x) = f(x-\theta)$ and when $f$ is analytic, see for instance \cite{ghosal2001entropies}. To the best of our knowledge it has not been extended to more general mixtures satisfying Assumption \textbf{A2}.  
\end{remark}

\subsection{Technical lemmas}

The following lemmas provide an expression for the characteristic function $\psi$ of $\mu(P)$ and subsequently establish that it is integrable. This fact is at the core of the proof of Theorem \ref{UBMDP1}.
\begin{lemma}
\label{lem_yamato}
(\cite{yamato1984characteristic})
Under Assumption A4, the characteristic function $\psi$ of $\mu(P)$ can be written as
\begin{equation}
    \psi(\boldsymbol z)=\mathbb E \left[\prod_{j=1}^\infty \varphi(\varpi_j\boldsymbol z)\right],\; \boldsymbol z\in\mathbb R^d
\end{equation}{}
\end{lemma}

\begin{lemma}
\label{lem_stickbreakingfinite}
For all $d\geq 1$, if $P=\sum_{j=1}^\infty \varpi_j\delta_{\theta_j}\sim DP(M,G_0)$ 
\begin{equation}
    \mathbb E_P\left[{(\sum_{j=1}^{d+1} \varpi_j^2)^{-\frac d2}}\right]<\infty
\end{equation}{}
\end{lemma}  
\textit{Proof: See Appendix}

\begin{lemma}\label{lem:densitymeanP}
If $P=\sum_{j=1}^\infty \varpi_j\delta_{\theta_j}\sim DP(M,G_0) $ where $G_0$ has a density with respect to Lebesgue measure and verifies 
$$ \int |\varphi(u)| du = C_{G_0}<\infty ,$$
with $\varphi$ the characteristic function of $G_0$ and $C_{G_0}$ a constant, then under Assumption A4 $\mu(P)$ has a bounded density with respect to the Lebesgue measure. 
\end{lemma}
\begin{proof}[Proof of Lemma \ref{lem:densitymeanP}]
Using Lemma \ref{lem_yamato}, the characteristic function $\psi$ of $\mu(P)$ is equal to
$$  \psi(\boldsymbol z)=\mathbb E \prod_{j=1}^\infty \varphi(\varpi_j\boldsymbol z),\; P =\sum_j \varpi_j \delta_{\theta_j}$$
where $\varphi$ is the characteristic function of $G_0$. Let $p_{(1)} = \max_{j} \varpi_j $ and note that, in particular, 
  $|\psi(\boldsymbol z)| \leq \mathbb E  |\varphi(p_{(1)}\boldsymbol z)| $. Hence, 
  $$ \int |\psi(\boldsymbol z)|d\boldsymbol z \leq \mathbb E_{p_{(1)}}\left[\frac{ 1 }{ p_{(1)}^d } \int |\varphi(u)| du \right] \leq C_{G_0} \mathbb E_{p_{(1)}}\left[\left( \frac{ \sum_{j=1}^r \varpi_j^2}{ r} \right)^{-d/2}\right],\text{ for some $r>0$} $$
Using Lemma \ref{lem_stickbreakingfinite}, with $r=d+1$, we obtain the result. 
\end{proof}

%%%%%%BEGINNING PROOOF WITH UNKNOWN SCALE sigma
\begin{comment}
In the next theorem we extend this result to location mixtures with a common but unknown scale, i.e. 
$$f(y) = f_{P, \sigma}(y) = \int_{\mathbb R^d } f((y-\theta)/\sigma)\sigma^{-d}dP(\theta), \quad f_0 = f_{P_0, \sigma_0}.$$
It is a surprisingly much more complex problem than the location mixture, because in this case we cannot use assumption \textbf{A2} which is crucial to our proof. Nevertheless the result still holds. 
For the sake of simplicity we restrict ourselves to the context of Dirichlet process mixture of Gaussians, other cases can be treated along those lines. Let $\varphi$ denotes the standard Gaussian density in $\mathbb R^{d_0}$ and $\varphi_\sigma$ be the density  of a centered Gaussian random variable with variance $\sigma^2 I_{d_0}$.
\begin{theorem}\label{UBMDP2}
Assume that $y_1, \cdots, y_n$ are iid $f_0 = \varphi_\sigma \ast P_0$ with $P_0 = \sum_{j=1}^{k_0}p_{j}^0 \delta_{\theta_j^0}$. Then there exists $t>0$ such that for all $\epsilon>0$
$$ \mathbb P_{f_0}\left( m_{DP}(\boldsymbol y ) > n^{-(k_0 + d_0k_0+t)/2}\right) =o(1)$$
\end{theorem}
\end{comment}

\section{Conclusion and prospectives}\label{sec:conclu}

This article addressed the difficult problem of evidence estimation for both finite and Dirichlet Process mixture models. 
For finite mixtures, we identified two methods, namely ChibPartition and SIS, that scale with the number of mixture components and/or with the number of observations, as opposed to classical algorithms. Furthermore, the adaptive SMC algorithm that we have presented shows good performance and can be used in the non-conjugate case. For the DPM, we noticed that there does not seem to exist any reference method widely used by practitioners. We benchmarked \cite{basu2003marginal}'s approach with other algorithms, including a method based on \cite{geyer1994estimating}'s reverse logistic regression that appears reliable and scales better with the number of observations $n$. 

An immediate application are goodness-of-fit tests that compare a parametric null to a nonparametric alternative. We formalized this procedure theoretically by establishing the consistence of the Bayes Factor for this kind of scenario, when the nonparametric alternative is a ‘strongly-identifiable' Dirichlet Process mixture model, which includes location mixture kernels.

An interesting research avenue would be to identify scalable evidence estimation techniques for non-conjugate Dirichlet Process mixtures. The RLR-Prior algorithm that we introduced can be a good first approach but suffers from a large variance when the number of observation increases.
As to the asymptotics of the marginal likelihood under a DPM, it would be interesting to extend Theorem \ref{UBMDP1} to the location-scale Dirichlet Process mixture model. 

\bibliographystyle{ims}
\bibliography{ref}

\appendix
\section{SIS sampling strategy for mixture models}
\label{sec:appendixSISFM}
The likelihood of a given partition of the data $\mathcal C(\boldsymbol z)=\{C_1(\boldsymbol z),\dots,C_K(\boldsymbol z)\}$ is given by 
\begin{equation}
    P(\boldsymbol y|\mathcal C(\boldsymbol z))=\prod_{k=1}^K p(\boldsymbol y|C_k(\boldsymbol z))=\prod_{k=1}^K\int_\Theta \prod_{z_i=k}p(y_i|\theta)G_0(d\theta):=\prod_{k=1}^K m(C_k(\boldsymbol z))
\end{equation}
For conjugate finite mixtures, this integral is easily computable in closed form.

Notice that the parameters $\boldsymbol\theta$ and the weights $\boldsymbol\eta$ are independent a posteriori and their respective augmented posteriors are given by
\begin{equation}
    P(d\theta_k|\boldsymbol y, \boldsymbol z)=\frac{\prod_{z_i=k}p(y_i|\theta_k)G_0(d\theta_k)}{\int_\Theta\prod_{z_i=k}p(y_i|\theta)G_0(d\theta)}=\frac{\prod_{z_i=k}p(y_i|\theta_k)G_0(d\theta_k)}{m(C_k(\boldsymbol z))}
\end{equation}
and
\begin{equation}
    P(\boldsymbol\eta|\boldsymbol z)=\mathcal D(n_1(\boldsymbol z)+\alpha_1,\dots,n_K(\boldsymbol z)+\alpha_K)
\end{equation}
where $n_k(\boldsymbol z)=\sum_{i=1}^n\1_{\{z_i=k\}},\; k=1,\dots,K$, provided the weights $\boldsymbol\eta$ are given a Dirichlet prior $\mathcal D(\alpha_1,\dots,\alpha_K)$.\\
The weights $w(\boldsymbol z,\boldsymbol y)$ can be derived by noticing that the posterior predictive distribution for all $y_i,$ where $i\geq2$ is given by
\begin{align*}
    p(y_i|\boldsymbol z_{1:i-1},\boldsymbol y_{1:i-1})&=\int\int p(y_i|\boldsymbol\eta,\boldsymbol\theta)\prod_{k=1}^KP(d\theta_k|\boldsymbol y_{1:i-1},\boldsymbol z_{1:i-1}) P(d\boldsymbol\eta|\boldsymbol z_{1:i-1})\\
    &=\sum_{k=1}^K\int \int \eta_k p(y_i|\theta_k)p(\theta_k|\boldsymbol y_{1:i-1},\boldsymbol z_{1:i-1}) d\boldsymbol\theta P(d\boldsymbol\eta|\boldsymbol z_{1:i-1})\\
    &=\sum_{k=1}^K \frac{m(C_k(\boldsymbol z_{1:i-1})\cup \{y_i\})}{m(C_k(\boldsymbol z_{1:i-1}))}\int \eta_k P(d\boldsymbol\eta|\boldsymbol z_{1:i-1})\\
    &=\sum_{k=1}^K \frac{m(C_k(\boldsymbol z_{1:i-1})\cup \{y_i\})}{m(C_k(\boldsymbol z_{1:i-1}))}\frac{n_k(\boldsymbol z_{1:i-1})+\alpha_k}{\sum_{k=1}^K n_k(\boldsymbol z_{1:i-1})+\alpha_k}
\end{align*}
The expression above determines completely the quantity of interest $w(\boldsymbol z,\boldsymbol y)$.\\

The allocation $z_i$ where $i\geq2$ should be drawn sequentially from the categorical distribution 
$$p(z_i=k|\boldsymbol y_{1:i},\boldsymbol z_{1:i-1})\propto \frac{m(C_k(\boldsymbol z_{1:i-1})\cup \{y_i\})}{m(C_k(\boldsymbol z_{1:i-1}))}\frac{n_k(\boldsymbol z_{1:i-1})+\alpha_k}{\sum_{k=1}^K n_k(\boldsymbol z_{1:i-1})+\alpha_k},\text{ for } k=1,\dots,K$$

\section{Additional lemma \label{sec:pr:LBf1}}
\begin{lemma}\label{lem:Constant}
Let $P=\sum_{i=1}^\infty \varpi_i\delta_{\theta_i}$ be a realisation of the Dirichlet Process. Define
\begin{equation}
\begin{split}
    \Delta(P)&=  \int_{A_0}f_\theta dP(\theta)  + 
    \sum_{j=1}^{k_0} \left[p_j-p_j^0\right] f_{\theta_j^0} +  \int_{B(\theta_j^0, \epsilon)}\left(\theta - \theta_j^0\right)^T\nabla f_{\theta_j^0}dP(\theta)  \\ 
    & \quad + \frac{ 1 }{2}\int_{B(\theta_j^0, \epsilon)}\left(\theta - \theta_j^0\right)^TD^2 f_{\theta_j^0}\left(\theta - \theta_j^0\right)dP(\theta)
    \end{split}
\end{equation}
where for all $j=1,\dots,k_0,$
\begin{equation}\label{DPweights}
    p_j=\sum_{i:\theta_i\in B(\theta_j^0,\epsilon)}\varpi_i
\end{equation}
for some $\epsilon>0$. Then there exists a constant $c(f_0)$ depending only on $f_0$ such that
\begin{equation} \label{LBf1}
\begin{split}
   \| \Delta(P)\|_1 &\geq c(f_0)\left[ P(A_0)  + 
    \sum_{j=1}^{k_0} \left|p_j-p_j^0\right|+ \left\|\int_{B(\theta_j^0, \epsilon)}\left(\theta - \theta_j^0\right) dP(\theta)\right\|+\right.\\
    &\quad\left.
   \int_{B(\theta_j^0, \epsilon)}\left\|\theta - \theta_j^0\right\|^2 dP(\theta)\right] \\
   &:= c(f_0)N(P)
   \end{split}
\end{equation}
\end{lemma}
\noindent\textit{Proof.}\\
Let 
\begin{align*}
    \tilde{\Delta}(P) &= \frac{ \Delta(P) }{ N(P)}\\
     & = \left\| \alpha_0 \int_{A_0}f_\theta(x) d\nu_0(x)  + 
    \sum_{j=1}^{k_0} \alpha_j f_{\theta_j^0} +  \beta_j^T\nabla f_{\theta_j^0}  + \frac{ 1 }{2}tr(D^2 f_{\theta_j^0}\gamma_j) \right\|_1
\end{align*} 
where 
\begin{equation*}
\begin{split}
    \alpha_0  &= \frac{P(A_0) }{ N(P)} , \quad \nu_0(d\theta) = \frac{ P(d\theta)\1_{A_0} }{ P(A_0)} , \quad \alpha_j = \frac{ p_j-p_j^0 }{ N(P)} \\
    \beta_j &= \frac{1}{ N(P)  }  \int_{B(\theta_j^0,\epsilon)}(\theta - \theta_j^0) dP(\theta), \quad 
    \gamma_j =  \frac{ 1 }{2N(P)}\int_{B(\theta_j^0, \epsilon)}(\theta - \theta_j^0) (\theta - \theta_j^0)^T dP(\theta)
\end{split}
\end{equation*}
and assume that \eqref{LBf1} does not hold. Then there exists a sequence $\alpha^m, \beta^m, \gamma^m, \nu_0^m$ along which $   \tilde \Delta(P)$ goes to zero. 
Note that by construction $tr(\gamma_j) = \int_{B(\theta_j^0, \epsilon)} \| \theta -\theta_j^0\|^2 dP(\theta)$ and that 
$\alpha^m, \beta^m, \gamma^m$ belong to a compact set so that there exists a sub-sequence  which is convergent to some value $\alpha, \beta, \gamma$. Similarly, $\nu_0^m$ is a sequence of measure with mass bounded by 1, so it converges vaguely to a sub-probability measure $\nu_0$ on $A_0$ along a subsequence (also denoted $\nu_0^m$) and since for all $x$, $f_\theta(x) $ is continuous in $\theta$ and converges to $0$ on the boundary of $\Theta$, 
$$\int_{A_0}f_\theta(x) d\nu_0^m \rightarrow \int_{A_0}f_\theta(x) d\nu_0, \quad \forall x$$
Now $\int_{A_0}f_\theta(x) d\nu_0^m \leq \sup_\theta \|f_\theta(\cdot)\|_\infty $ hence on any compact subset $B$ of $\mathbb R^{d}$, 
$$\|\1_{B}\int_{A_0}f_\theta(x) (d\nu_0^m - d\mu_0)(\theta) \|_1 =o(1)$$ so that for all compact $B$, 
at the limit, for all $x\in B$
$$  \alpha_0 \int_{A_0}f_\theta(x) d\nu_0(\theta)  + 
    \sum_{j=1}^{k_0} \alpha_j f_{\theta_j^0}(x) +  \beta_j^T\nabla f_{\theta_j^0}(x)  + \frac{ 1 }{2}tr(D^2 f_{\theta_j^0}(x)\gamma_j) =0$$
    since the relation is true for all $B$ it is true for all $x\in \mathbb R^d$. 
    The strong identifiability assumption \textbf{A2} implies that $\alpha=0$, $\beta= 0$ and $\gamma=0$, which is not possible since $\sum_j |\alpha_j| + |\beta_j| + |\gamma_j| =1$. Hence  \eqref{LBf1} is valid. 
\section{Proof of Lemma \ref{lem_stickbreakingfinite}}
We first notice that
\begin{align}\label{jacob}
    \begin{cases}{}
    V_1=\varpi_1\\
    V_2(1-V_1)=\varpi_2\\
    \vdots\\
    V_{d+1}(1-V_d)\dots (1-V_1)=\varpi_{d+1}
    \end{cases}
    \Leftrightarrow
    \begin{cases}{}
    V_1=\varpi_1\\
    V_2=\frac{\varpi_2}{1-\varpi_1}\\
    \vdots\\
    V_{d+1}=\frac{\varpi_{d+1}}{1-\sum_i^d \varpi_i}
    \end{cases}
\end{align}{} 
where $V_i\overset{i.i.d}{\sim}\mathcal{B}(1,M)$.
Hence,
\begin{equation*}
  \begin{split}
    I=&\mathbb E\left[\frac{1}{(\sum_{j=1}^{d+1} \varpi_j^2)^{\frac d2}}\right]\\
    =&\left(\frac{\Gamma(M)}{\Gamma(1+M)}\right)^{d+1}\int_{[0,1]^{d+1}}\frac{(1-V_1)^{M-1}\dots(1-V_{d+1})^{M-1}}{(V_1^2+\dots+V_{d+1}^2(1-V_d)^2\dots(1-V_1)^2)^{d/2}}
    dV_1\dots dV_{d+1}
\end{split}  
\end{equation*}

Let $B_\varepsilon=\{(V_1,\dots,V_{d+1}):\exists i\in\{1,\dots,d+1\},\;\varpi_i>\epsilon\}$ and fix $\varepsilon=1/(2(d+1))$. Then,

\begin{equation*}
  \begin{split}
    I\leq&\left(\frac{\Gamma(M)}{\Gamma(1+M)}\right)^{d+1}\left (C\varepsilon^{-d}+\int_{[0,1]^{d+1}\cap B_\varepsilon^c}\frac{(1-V_1)^{M-1}\dots(1-V_{d+1})^{M-1}}{(V_1^2+\dots+V_{d+1}^2(1-V_d)^2\dots(1-V_1)^2)^{d/2}}
    dV_1\dots dV_{d+1}\right )
\end{split}  
\end{equation*}
for some constant $C>0$.
We then use the change of variable given by (\ref{jacob}) which Jacobian is 
\begin{equation*}
    \det(J)=\frac{1}{1-\varpi_1}\times\frac{1}{1-\varpi_1-\varpi_2}\times\dots\times\frac{1}{1-\sum_i^d \varpi_i}
\end{equation*}{}
Hence,
\begin{equation*}
    I\leq\left(\frac{\Gamma(M)}{\Gamma(1+M)}\right)^{d+1}\left(C\varepsilon^{-d}+\int_{\Delta_{d+1}\cap B_\varepsilon^c}\frac{1}{(\sum_{j=1}^{d+1}\varpi_j^2)^\frac d2}\times\frac{\left(1-\sum_{j=1}^{d+1}\varpi_j\right)^{M-1}}{(1-\varpi_1)\dots(1-\sum_{j=1}^d \varpi_j)}d\varpi_1\dots d\varpi_{d+1}\right)
\end{equation*}
where $\Delta_{d+1}=\{\boldsymbol q \in \mathbb R^{d+1}: 0\leq \varpi_j\leq1, \forall j$ and $\sum_j \varpi_j\leq1\}$.\\

\begin{align*}
    I&\leq \left(\frac{\Gamma(M)}{\Gamma(1+M)}\right)^{d+1}\left(C  \varepsilon^{-d}+2^d\int_{\Delta_{d+1}\cap B_\varepsilon^C}
    \frac{\left(1-\sum_{j=1}^{d+1}\varpi_j\right)^{M-1}}{(\sum_{j=1}^{d+1}\varpi_j^2)^\frac d2}d\varpi_1\dots d\varpi_{d+1}\right)
    \intertext{If $M\geq 1$}
    I&\leq \left(\frac{\Gamma(M)}{\Gamma(1+M)}\right)^{d+1}\left(\varepsilon^{-d}C+2^d\int_{\Delta_{d+1}}\frac{1}{(\sum_{j=1}^{d+1}\varpi_j^2)^\frac d2}d\varpi_1\dots d\varpi_{d+1}\right)
    \intertext{If $M< 1$}
    I&\leq \left(\frac{\Gamma(M)}{\Gamma(1+M)}\right)^{d+1}\left(\varepsilon^{-d}C+2^{d+1-M}\int_{\Delta_{d+1}}\frac{1}{(\sum_{j=1}^{d+1}\varpi_j^2)^\frac d2}d\varpi_1\dots d\varpi_{d+1}\right)
\end{align*}
Let $J=\int_{\Delta_{d+1}}
    1/(\sum_{j=1}^{d+1}\varpi_j^2)^\frac d2d\varpi_1\dots d\varpi_{d+1}$. We shall prove that $J$ is finite.\\
    
Using the hyperspherical change of variable
\begin{align*}{}
\begin{cases}
\varpi_1=r\cos \theta_1\\
\varpi_2=r\sin \theta_1\cos\theta_2\\
\vdots\\
\varpi_{d}=r\sin\theta_1\dots\cos\theta_d\\
\varpi_{d+1}=r\sin\theta_1\dots\sin\theta_d
\end{cases}{}
\end{align*}
where $\theta_1,\dots,\theta_{d-1}\in[0,\pi]$ and $\theta_d\in[0,2\pi)$.
The Jacobian of this transformation is bounded by $r^d$, and noticing that $r=\left(\sum_{j=1}^{d+1} \varpi_j^2\right)^\frac 12$, we get
\begin{align*}
    J\leq \int_0^1\int_0^{2\pi}\dots\int_0^\pi\frac{r^d}{r^{d}} d\theta_1\dots d\theta_d dr
\end{align*}
Hence,
\begin{equation*}
    J\leq 2\pi\left(\frac\pi2\right)^{d-1}
\end{equation*}
and $$\mathbb E\left[\frac{1}{(\sum_{j=1}^{d+1} \varpi_j^2)^{\frac d2}}\right]<\infty$$ 

\section{Proof of Theorem \ref{UBMDP1}}
To prove Theorem \ref{UBMDP1}, we  prove that the associated posterior concentrates in $L_1$ norm at the rate $(\log n)^q/\sqrt{n}$  for some $q>0$ and then we bound from above $\Pi( \{ P: \|f_P - f_0\|_1 \leq (\log n)^q/\sqrt{n}\})$. More precisely we  first prove that 
\begin{equation}\label{my2}
\mathbb P_{f_0}\left( \int_{\|f_P - f_0\|_1> \delta_n}e^{\ell_n(P) - \ell_n(P_0)}d\Pi(P) >  n^{-D(k_0)-t}\right)   = o(1) 
\end{equation}
where $ \delta_n =   (\log n)^q/\sqrt{n}$, $ D(k_0) = (k_0-1 + dk_0)/2$ and $\ell_n(P) = \log f_P(\mathbf y)$ is the log-likelihood.

Then, to bound from above $\int_{\|f_P - f_0\|_1\leq \delta_n}e^{\ell_n(P) - \ell_n(P_0)}d\Pi(P) $ we use a simple Markov inequality:
\begin{equation}\label{my1}
\mathbb P_{f_0}\left( \int_{\|f_P - f_0\|_1\leq \delta_n}e^{\ell_n(P) - \ell_n(P_0)}d\Pi(P) >  n^{-D(k_0)-t}\right)\leq 
 n^{D(k_0)+t} \Pi (\|f_P - f_0\|_1\leq \delta_n )
\end{equation}
and we can conclude the proof as soon as we bound $\Pi (\|f_P - f_0\|_1\leq \delta_n )$. This is the most difficult part of the proof and is done in Section \ref{sec:L1neighbour}.

\subsection{Proof of $L_1$-concentration rate \eqref{my2}}
Throughout the proof $C$ denotes a generic constant depending only on $f_0$. 
Let $\eta>0$ be arbitrarily small, and 
$$\mathcal F_n = \{ f_P;\, P(\Theta_n^c)\leq \eta \delta_n \}, \quad \Theta_n = \{ \|\theta\| \leq (\log n)^a \},$$ for some $a>0$. 
To prove \eqref{my2}, we first prove that the covering number of each slice  
$$ \mathcal F_{n,\ell}  = \{ f_P \in \mathcal F_n; \|f_P - f_0\|_1 \in (\ell \delta_n, (\ell+1)\delta_n) \}$$
by $L_1$ balls of radius $\zeta \ell \delta_n$ is bounded by 
\begin{equation}\label{entropy:slice}
\mathcal N_{n,\ell} \lesssim e^{C(\log n)^{ad} \log (1/\eta)}.
\end{equation} 
Using Lemma \ref{lem:Constant}, we have that for all $f_P\in \mathcal F_{n,\ell}  $, if $\ell \delta_n \leq \epsilon_0$ for some $\epsilon_0>0$ small enough, 
$$ P(A_0) + \sum_{j=1}^{k_0}\left[ |p_j-p_j^0|+ \| \mu(P_j) -\theta_j^0\| + \int_{B(\theta_j^0, \epsilon) } \|\theta - \theta_j^0\|^2 dP(\theta) \right] \leq (\ell +1) \delta_n/c(f_0)$$ 
where $\mu(P_j):=\int_{B(\theta_j^0,\epsilon)}\theta dP(\theta)$.\\
Let $\xi >0$, we consider a covering of $A_{0,n}=A_0\cap\Theta_n$ with balls $U_{i,0} $ with center $\vartheta_i$ and radius $\xi$ and $N_0(\xi)$, the number of such balls, is bounded by $(C |\Theta_n|/\xi)^d$. Define $\tilde U_{1,0}=U_{1,0}$ and $\tilde U_{i,0}=U_{i,0}\setminus\cup_{j=1}^{i-1}  U_{j,0}$. Consider $P,P'$ such that $f_P, f_P'\in \mathcal F_{n,\ell}  $, 
%$|P(A_0) - P'(A_0)|\leq \varepsilon (\ell+1)\delta_n$
\begin{equation}\label{PA0}
\begin{split}
    \int \left| \int_{A_0} f_\theta(x) (dP-dP')(\theta) \right|dx &\leq \int \left|\sum_{i=1}^{N_0(\xi)}\left\{\int_{\tilde U_{i,0}}f_\theta(x)-f_{\vartheta_i}(x)(dP-dP')(\theta)\right.\right.\\
    &\qquad \left.\left. +\int_{\tilde U_{i,0}}f_{\vartheta_i}(x)(dP-dP')(\theta)\right\}\right|\\
    &\leq\int \sum_{i=1}^{N_0(\xi)}\left\{\int_{\tilde U_{i,0}}|f_\theta(x)-f_{\vartheta_i}(x)|dP(\theta)\right.\\
    &\qquad \left. +\int_{\tilde U_{i,0}}|f_\theta(x)-f_{\vartheta_i}(x)|dP'(\theta)+f_{\vartheta_i}(x)\left|\int_{\tilde U_{i,0}}d(P-P')(\theta)\right|\right\}\\
    &\leq||H_1(x)||_1\xi\left[P(A_0)+P'(A_0)\right]+\sum_{i=1}^{N_0(\xi)}\left|P(\tilde U_{i,0})-P'(\tilde U_{i,0})\right|\\
    &\leq C\xi (\ell+1)\delta_n+\sum_{i=1}^{N_0(\xi)} \left|P(\tilde U_{i,0})-P'(\tilde U_{i,0})\right|
    \end{split}
\end{equation}
Hence if for all $i\leq N_0(\xi)$, $|P(\tilde U_{i,0})-P'(\tilde U_{i,0})|\leq \xi / N_0(\xi)$,
\begin{equation}\label{PA0}
\begin{split}
    \left\| \int_{A_0} f_\theta(\cdot) (dP-dP')(\theta) \right\|_1 &\leq C\xi \ell \delta_n .
    \end{split}
\end{equation}
Moreover for each $j=1,\dots, k_0$
\begin{equation*}%\label{PA0}
\begin{split}
     \int_{B(\theta_j^0, \epsilon)} f_\theta(x) (dP-dP')(\theta)  &= \int_{B(\theta_j^0, \epsilon)}\left\{f_{\theta_j^0}(x)+(\theta-\theta_j^0)^T\nabla f_{\theta_j^0}(x)+\frac12(\theta-\theta_j^0)^TD^2f_{\theta_j^0}(x)(\theta-\theta_j^0)\right.\\
     &\qquad \left. + \mathcal O\left(\|\theta- \theta_j^0\|^{2+\delta}\right) \right\}\left[dP-dP'\right](\theta)\\
     &=\left[f_{\theta_j^0}(x)-{\theta_j^0}^T\nabla f_{\theta_j^0}(x)\right]\left(p_j-p_j'\right) + \left[\mu(P_j)-\mu(P_j')\right]^T \nabla f_{\theta_j^0}(x)\\
     & \qquad + \frac12 tr\left[D^2f_{\theta_j^0}(x )\int (\theta- \theta_j^0)(\theta- \theta_j^0)^Td(P_j-P_j')(\theta)\right]   \\
     & \qquad + \mathcal O\left(\int \|\theta- \theta_j^0\|^{2+\delta} d(P_j-P_j')(\theta)\right) \\
     \end{split}
\end{equation*}
where $p$ and $p'$ are defined as in \eqref{DPweights}.\\
Hence if $|p_j-p_j' | \leq \xi (\ell+1)\delta_n$ ,
$|\mu(P_j)-\mu(P_j')| \leq \xi (\ell+1)\delta_n $ and 
$$\left\| \int (\theta- \theta_j^0)(\theta- \theta_j^0)^Td(P_j-P_j')(\theta) \right\|_F \leq \xi (\ell+1)\delta_n,$$
\begin{equation*}%\label{PA0}
\begin{split}
    &\left\| \int_{B(\theta_j^0, \epsilon)} f_\theta(\cdot) (dP-dP')(\theta) \right\|_1 \\
    &\leq \xi (\ell+1)\delta_n\left( 1 + \|\theta_j^0\|\| H_1(\cdot)\|_1+\|H_1(\cdot)\|_1
     + \frac{ \|H_2(\cdot)\|_1 }{2} \right)
    + \mathcal O( \epsilon^\delta (\ell+1)\delta_n )\\
    & \leq C(\epsilon^\delta+\xi) \ell \delta_n ).
     \end{split}
\end{equation*}
So that $\|f_P-f_{P'}\|_1 \leq C \xi\ell \delta_n$. \\

Since for all $j$, when $f_P \in \mathcal F_{n,\ell}$, $\mu(P_j) \in B(\theta_j^0, C\ell\delta_n)$, the covering number for $\{ \mu(P_j); \, f_P \in \mathcal F_{n,\ell}\}$  by balls of radius is $\xi \ell \delta_n$ is bounded by $(C\xi)^{-d}$ and similarly the covering numbers for the terms $ \int (\theta- \theta_j^0)(\theta- \theta_j^0)^Td P_j(\theta)$ and $p_j$ are bounded respectively by 
$(C \xi )^{-d(d+1)/2}$ and $(C \xi )^{-1}$. 

This implies that by choosing $\xi$ and small enough, 
\begin{align*}
    \mathcal N_{n,\ell} & \lesssim \xi ^{-k_0[1+d+d(d+1)/2]}(N_0(\xi)/\xi)^{N_0(\xi)} \\
    & \leq \exp\left[\left(C|\Theta_n|/\xi\right)^d\log(|\Theta_n|/\xi)+\log(1/\xi)\right]\\
    & \leq \exp\left[\left(C|\Theta_n|/\xi\right)^d(\log(|\Theta_n|)+\log(1/\xi))\right]\\
    & \leq \exp\left[\left(C(\log n)^a/\xi\right)^d(\log((\log n)^a)+\log(1/\xi))\right]
\end{align*}
and since $|\Theta_n| \leq (\log n)^a$ for some $a>0$, 
\eqref{entropy:slice} holds. 

Now if $\ell \delta_n > \epsilon_0$, we use a similar decomposition as in \eqref{PA0} but with $\Theta_n$ instead of $A_0$, which leads to a covering number bounded by 
$$\log \mathcal N_{n,\ell} \lesssim  (\log n)^{ad} \log (1/\varepsilon) .$$

Thus similarly to Theorem 7.1 of \cite{ghosal2000convergence}, 
there exists ($L_1$) tests $\phi_{n,\ell}$ such that 
\begin{align*}
\mathbb E_0(\phi_n) &\leq \sum_{\ell\geq \ell_0}\mathcal N_{n,\ell}e^{- c_1 n\ell^2 \delta_n^2} =o(1) \\
\sup_{\|f_p - f_0\|> \ell_0\delta_n, \in \mathcal F_n} \mathbb E_{f_P}(1-\phi_n)&\leq e^{- c_1 n\ell_0^2 \delta_n^2}
\end{align*}
if $\ell_0$ is large enough and $2q>ad$ .
Moreover from \cite{ghosal2000convergence},
$$ \Pi (P(\Theta_n^c) > \varepsilon \delta_n ) \lesssim \frac{ G_0(\|\theta\| \geq (\log n)^a) }{ \varepsilon\delta_n} \leq n^{- (D(k_0) + 2t)}$$
by choosing $a\tau >1$ large enough. 

We finally obtain that 
\begin{equation*}
\begin{split}
\mathbb P_{f_0}&\left( \int_{\|f_P - f_0\|_1> \delta_n}e^{\ell_n(P) - \ell_n(P_0)}d\Pi(P) >  n^{-D(k_0)-t}\right) \leq n^{-t} + e^{-c_1 \ell_0^2n \delta_n^2} n^{D(k_0)+t} =o(1)
\end{split}
\end{equation*}
for some $c_1>0$ if $q$ is large enough  and \eqref{my2} is proved.

\subsection{Upper bound on $\Pi (\|f_P - f_0\|_1\leq \delta_n )$. } \label{sec:L1neighbour}

In this section we prove that
\begin{equation}\label{my1}
\Pi (\|f_P - f_0\|_1\leq \delta_n ) = o(n^{-D(k_0) -t}).
\end{equation}
for some $t>0$. 
%Recall that $\Theta$ is assumed to be compact. 
Using a similar idea to Lemma 3.8 of \cite{gassiat2014local}, we choose $\epsilon \leq \min \{ \| \theta_i^0 - \theta_j^0\|/4; i \neq j\leq k_0\}$, so that the balls $B(\theta_i^0, \epsilon) $, $i\leq k_0$ are disjoint. 
We then write with $A_0 = (\cup_{i=1}^{k_0}B(\theta_i^0, \epsilon))^c$ and $p_j=P(B(\theta_j^0, \epsilon))$
\begin{equation*}
\begin{split}
    \|f_0 - f_P\|_1 &= \left\| \int_{A_0}f_\theta dP(\theta)  + 
    \sum_{j=1}^{k_0} \left[p_j-p_j^0\right] f_{\theta_j^0} +  \int_{B(\theta_j^0, \epsilon)}\left(\theta - \theta_j^0\right)^T\nabla f_{\theta_j^0}dP(\theta) \right. \\ 
    & \quad \left. + \frac{ 1 }{2}\int_{B(\theta_j^0, \epsilon)}\left(\theta - \theta_j^0\right)^TD^2 f_{\theta_j^0}\left(\theta - \theta_j^0\right)dP(\theta) + R \right\|_1 \\
    &:= \| \Delta_n(P) + R\|_1
\end{split}
\end{equation*}
where 
\begin{equation*}
\begin{split} 
\|R \|_1 &\leq \sum_{j=1}^{k_0}  \int_{B(\theta_j^0, \epsilon)}\left\|\sup_{\|\theta'-\theta_0^j\|\leq \|\theta-\theta_j^0\|}\left|\left(\theta - \theta_j^0\right)^T\left[D^2 f_{\theta'}-D^2 f_{\theta_j^0}\right]\left(\theta - \theta_j^0\right)\right|\right\|_1dP(\theta) \\
 & \leq d \sum_{j=1}^{k_0} \|H_3\|_1\int_{B(\theta_j^0, \epsilon)}\left\|\theta - \theta_j^0\right\|^{2 + \delta}dP(\theta).
\end{split}
\end{equation*}
By Lemma \ref{lem:Constant}, there exists $c(f_0)$ such that 
\begin{equation} 
\begin{split}
   \| \Delta_n(P)\|_1 &\geq c(f_0)\left[ P(A_0)  + 
    \sum_{j=1}^{k_0} \left|p_j-p_j^0\right|+ \left\|\int_{B(\theta_j^0, \epsilon)}\left(\theta - \theta_j^0\right) dP(\theta)\right\|+\right.\\
    &\quad\left.
   \int_{B(\theta_j^0, \epsilon)}\left\|\theta - \theta_j^0\right\|^2 dP(\theta)\right] \\
   &:= c(f_0)N_n(P)
   \end{split}
\end{equation}
Then if $ \|f_0 - f_P\|_1 \leq \delta_n$, then 
$$ N_n(P)c(f_0) \leq \delta_n +d \sum_{j=1}^{k_0} \|H_3\|_1\int_{B(\theta_j^0, \epsilon)}\|\theta - \theta_j^0\|^{2 + \delta}dP(\theta)$$
and choosing $\epsilon$ small enough
$$ N_n(P) \leq \frac{ 2\delta_n }{c(f_0)}:=\delta_n' .$$
We now bound from above 
 $\Pi(N_n(P) \leq \delta_n')$. Consider $P_j = P\1_{B(\theta_j^0, \epsilon)}/p_j$ and $p_0 = P(A_0)$, then
  under the Dirichlet Process prior $(p_0, \cdots, p_{k_0}); P_1, \cdots , P_{k_0}$ are mutually independent and $P_j \sim DP( MG_0(\cdot \cap B(\theta_j^0, \epsilon))$. Hence
  \begin{equation*} 
  \begin{split}
  \Pi(N_n(P) \leq \delta_n') & \lesssim \delta_n^{MG(A_0)}\delta_n^{k_0-1}\prod_{j=1}^{k_0}\Pi\left( \left\|\int (\theta - \theta_j^0)dP_j(\theta)\right\|\leq \delta_n'\right) 
  \end{split}
    \end{equation*}
Let $\Pi_j$ be the marginal prior distribution of $\int \theta dP_j(\theta)$, then 
$$ \Pi_j( B(\theta_j^0, \delta_n') )\leq  \sup_{\|x - \theta_j^0\|\leq \delta_n'} \pi_j(x)(\delta_n')^{d} \lesssim \delta_n^{d}$$
as soon as $\Pi_j$ has a bounded density $\pi_j$.  Using Lemma \ref{lem:densitymeanP}, we know that $\Pi_j$ has a bounded density $\pi_j$ and \eqref{my1} is proved.

\section{Choice of tuning parameters for the numerical experiments}

\subsection{Figure \ref{fig:boxplotFMgalaxy}}

\begin{table}[H]
    \centering
    \label{tab:hyperParamFM}
    \begin{tabular}{c|cccc}
\hline
\multicolumn{5}{c}{Tuning parameters}\\
&$K=3$ &$K=5$& $K=6$ & $K=8$\\
\hline
\hline
Chib &$T=10^5$& $T=10^5$ & $T=2\cdot10^5$ & $T=3\cdot10^5$ \\
&$burnIn=10^4$&$burnIn=10^4$&$burnIn=2\cdot10^4$&$burnIn=3\cdot10^4$\\
\hline
Chib Permutation &$T=10^5$& $T=10^5$ &\multirow{2}{*}{-}& \multirow{2}{*}{-} \\
&$burnIn=10^4$&$burnIn=10^4$&&\\
\hline
Adaptive SMC &$N=10000$& $N=10000$ & $N=10000$ &$N=10000$\\
&M=10&M=10&M=10&M=10\\
\hline
Bridge Sampling &$T_1=12\cdot10^3$& $T_1=12\cdot10^3$ &  \multirow{4}{*}{-} &  \multirow{4}{*}{-}\\
&$T_2=12\cdot10^3$&$T_2=12\cdot10^3$&&\\
&$T_0=100$&$T_0=100$&&\\
&$burnIn=5000$&$burnIn=5000$&&\\
\hline
ChibPartition &$T=10^5$& $T=10^5$ & $T=2\cdot10^5$ & $T=3\cdot10^5$\\
&$burnIn=10^4$&$burnIn=10^4$&$burnIn=2\cdot10^4$&$burnIn=3\cdot10^4$\\
\hline
SIS &$T=10^3$& $T=6\cdot10^3$ & $T=7\cdot10^3$&$T=10^4$\\
\hline
Arithmetic Mean &$T=3\cdot10^6$& $T=3\cdot10^6$ & $T=3\cdot10^6$&$T=4\cdot10^6$\\
\end{tabular}
\caption{Choice of tuning parameters for Figure \ref{fig:boxplotFMgalaxy}}
\end{table}

\subsection{Figure \ref{fig:boxplotSynthData}}
     \begin{table}[H]
    \centering
    \begin{tabular}{ccc|cc}
\hline
&\multicolumn{2}{c}{$n=1000$}&\multicolumn{2}{c}{$n=2000$}\\
&$K=3$ &$K=13$& $K=3$ & $K=13$\\
\hline
\hline
Adaptive SMC &$N=2\cdot10^4$& $N=2\cdot10^4$ & $N=2\cdot10^4$ &$N=2\cdot10^4$\\
&M=10&M=10&M=10&M=10\\
\hline
ChibPartition &$T=10^5$& \multirow{2}{*}{-} &\multirow{2}{*}{-} &\multirow{2}{*}{-}\\
&$burnIn=10^4$&&&\\
\hline
Bridge Sampling &$T_1=12\cdot10^3$& \multirow{4}{*}{-} &$T_1=12\cdot10^3$ &\multirow{4}{*}{-}\\
&$T_2=12\cdot10^3$&&$T_2=12\cdot10^3$&\\
&$T_0=100$&&$T_0=100$&\\
&$burnIn=5000$&&$burnIn=5000$&\\
\hline
SIS &$T=2\cdot10^3$& $T=2\cdot10^3$ & $T=10^4$&$T=10^4$\\
\end{tabular}
    \caption{Choice of tuning parameters for Figure \ref{fig:boxplotSynthData}.}
    \label{tab:hyperParamSynthFM}
\end{table}

\subsection{Figure \ref{fig:galaxyDPM}}
\begin{table}[h]
    \centering
    
    \begin{tabular}{c|ccc}
\hline
\multicolumn{4}{c}{Tuning parameters}\\
&$n=6$ &$n=36$& $n=82$\\
\hline
\hline
Chib &$T_1=3\cdot 10^4$& $T_1=5\cdot 10^4$ & $T_1=10^5$  \\
&$burnIn=2\cdot10^3$&$burnIn=5\cdot10^3$&$burnIn=10^4$\\
&$T_2=2\cdot10^3$&$T_2=2\cdot10^3$&$T_2=2\cdot10^3$\\
\hline
RLR-SIS &$T_1=3\cdot 10^4$& $T_1=5\cdot 10^4$ &$T_1= 10^5$ \\
&$burnIn=2\cdot10^3$&$burnIn=5\cdot10^3$&$burnIn=10^4$\\
&$T_2=2\cdot10^3$&$T_2=2\cdot10^3$&$T_2=2\cdot10^3$\\
\hline
RLR-Prior &$T_1=3\cdot 10^4$&$T_1=5\cdot 10^4$&  $T_1= 10^5$ \\
&$burnIn=2\cdot10^3$&$burnIn=5\cdot10^3$&$burnIn=10^4$\\
&$T_2=2.8\cdot10^4$&$T_2=4.5\cdot10^4$&$T_2=9\cdot10^4$\\
\hline
Harmonic Mean &$T=2\cdot10^4$& $T_1=5\cdot10^4$ & $T_1=10^5$\\
&$burnIn=2\cdot 10^3$&$burnIn=5\cdot 10^3$&$burnIn=10^4$\\
\hline
Arithmetic Mean &$T=3\cdot10^4$& $T=2\cdot10^5$ & $T=5\cdot10^5$
\end{tabular}
\caption{Choice of tuning parameters for Figure \ref{fig:galaxyDPM}}
\label{tab:hyperParamDPM}
\end{table}

\end{document}